\documentclass[aps,amsmath,amssymb,twocolumn,prl,superscriptaddress]{revtex4-2}

\usepackage{float}
\usepackage{graphicx,tikz}% Include figure files
\usepackage{dcolumn}% Align table columns on decimal point
\usepackage{bm}% bold math
\usepackage{epsfig,hyperref,amsthm}
\usepackage{mathtools}
\usepackage{color,fancybox}
\usepackage{colordvi}
\usepackage{relsize}
\usepackage{accents}
\usepackage{wasysym} 
\usepackage{enumitem}
\usepackage{mathtools,slashed}
\usepackage{times}
\usepackage{newtxmath}

\hypersetup{
	colorlinks=true,
	linkcolor=CitingColor,
	citecolor=CitingColor,
	urlcolor=CitingColor
}

\DeclareMathAlphabet\mathbfcal{OMS}{cmsy}{b}{n}

\newcommand{\ket}[1]{\ensuremath{|#1\rangle}}
\newcommand{\bra}[1]{\ensuremath{\langle #1|}}

\newcommand{\proj}[1]{\ket{#1}\!\bra{#1}}
\newcommand{\be}{\begin{equation}}
\newcommand{\ee}{\end{equation}}
\newcommand{\ba}{\begin{eqnarray}}
\newcommand{\ea}{\end{eqnarray}}

\newcommand{\norm}[1]{\left\|#1\right\|}

\newcommand{\id}{\mathbb{I}}
\newcommand{\Wdiff}{\Delta}
\newcommand{\WdiffSA}{\overline{\Delta}_\gamma}
\newcommand{\Ydual}{{Y}}
\newcommand{\SRvar}{{\color{black}\mu}}

\newcommand{\C}{\mathcal{C}}

\newtheorem{result}{Result}

\newtheorem{lemma}[result]{Lemma}

\newtheorem{proposition}[result]{Proposition}
\newtheorem{question}{Question}
\newtheorem{fact}[result]{Fact}

\definecolor{nred}{rgb}{0.9,0.1,0.1}
\definecolor{nblack}{rgb}{0,0,0}
\definecolor{nblue}{rgb}{0.2,0.2,0.8}
\definecolor{ngreen}{rgb}{0.2,0.5,0.2}
\definecolor{ublue}{rgb}{0,0,0.5}

\definecolor{pur}{rgb}{0.75,0,0.75}
\definecolor{nngrn}{rgb}{0,0.5,0.5}
\definecolor{CitingColor}{rgb}{0,0.3,1}

\newcommand{\blu}{\color{nblue}}

\newcommand{\CY}[1]{{\color{black}#1}}
\newcommand{\CYnew}[1]{{\color{black}#1}}
\newcommand{\CYtwo}[1]{{\color{black}#1}}
\newcommand{\CYthree}[1]{{\color{black}#1}}

\newcommand{\RFone}[1]{{\color{black}#1}}
\newcommand{\RFtwo}[1]{{\color{black}#1}}

\usepackage[capitalize]{cleveref}

\begin{document}
\title{Thermodynamic Approach to Quantifying Incompatible Instruments}

\author{Chung-Yun Hsieh}
%\email{chung-yun.hsieh@bristol.ac.uk}
\affiliation{H. H. Wills Physics Laboratory, University of Bristol, Tyndall Avenue, Bristol, BS8 1TL, United Kingdom}

\author{Shin-Liang Chen}
\email{shin.liang.chen@email.nchu.edu.tw}
\affiliation{Department of Physics, National Chung Hsing University, Taichung 40227, Taiwan}
\affiliation{Physics Division, National Center for Theoretical Sciences, Taipei 106319, Taiwan}
\affiliation{\CYtwo{Center for Quantum Frontiers of Research \& Technology (QFort), National Cheng Kung University, Tainan 701, Taiwan}}

\date{\today}

\begin{abstract} 
We consider a thermodynamic framework to quantify instrument incompatibility via a resource theory {\em subject to} thermodynamic constraints. We use the minimal {\em thermalisation time} needed to erase incompatibility's signature to measure incompatibility. Unexpectedly, this time value is equivalent to incompatibility advantage in a work extraction task. Hence, both thermalisation time and extractable work can directly quantify instrument incompatibility. Finally, we show that incompatibility signatures must vanish in non-Markovian thermalisation.
\end{abstract}

\maketitle

% \section{Introduction}
% {\bf\em Introduction.---}
Uncertainty principle is one of the most profound aspects of quantum theory~\cite{Busch2014RMP}. It results from the fact that quantum observables, in general, are not commuting, and measuring one physical property can unavoidably forbid us from knowing anything about the other. Namely, there are quantum properties that cannot be {\em jointly measured} via a single quantum device, as termed {\em incompatible}~\cite{Otfried2021Rev}. Crucially, quantum theory's incompatible nature is for more than just measurement devices---two state ensembles can be incompatible~\cite{Hsieh-IP} via the phenomenon called quantum steering~\cite{UolaRMP2020,Cavalcanti2016}, and quantum channels can be incompatible as they are not always simultaneously implementable in broadcast scenarios~\cite{ScaraniRMP2005,Hsieh2022PRR,Haapasalo2021}. It turns out that different types of incompatible quantum devices are potential resources in various operational tasks, such as, but not limited to, one-sided device-independent quantum information tasks (see, e.g., Refs.~\cite{UolaRMP2020,Cavalcanti2016}), state/channel discrimination~\cite{Skrzypczyk2019,Takagi2019,Uola2019PRL,Carmeli2019PRL,Hsieh2022PRR,Designolle2019,Hsieh2023-2}, state/channel exclusion~\cite{Hsieh-IP,Ducuara2020}, quantum programmability~\cite{Buscemi2020PRL,Buscemi2023}, and encryption~\cite{Hsieh-IP}. For a better global view, it is necessary to have a mathematical language unifying different types of incompatibility. This thus initiates the study of quantum instruments and their incompatibility~\cite{Ku2018PRA,Ku2021PRR,Ku2022PRXQ,Mitra2022,Mitra2023,Buscemi2023,Heinosaari2016,Heinosaari2014,Ji2023}.

So far, most discussions have mainly focused on incompatible instruments from the quantum-information perspective. A physically relevant question is how to quantify this quantum feature from a {\em thermodynamic} point of view. For example, thermodynamic approaches to understanding/quantifying information transmission~\cite{Landauer1961,Hsieh2021PRXQ,Hsieh2022}, conditional entropy~\cite{delRio2011}, and quantum correlation~\cite{Oppenheim2002,Perarnau-Llobe2015,JiPRL2022,BeyerPRL2019,ChanPRA2022} have provided novel insights and advanced our understanding of the relation between thermodynamics and quantum information. Also, \CYtwo{quantum theory's} non-commuting signatures have been studied in the context of quantum heat engines~\cite{KosloffPRE2002,FeldmannPRE2006} and conservation laws~\cite{LostaglioNJP2017,Majidy2023,YungerHalpern2016NC,Guryanova2016NC}, showing \CYtwo{incompatibility's} important role in thermodynamics. Nevertheless, a thermodynamic framework for understanding incompatible instruments is still missing in the literature.

This work fills this gap by considering a natural framework for quantifying and characterising incompatible instruments via thermalisation. The idea is to analyse instruments' ability to drive the system out of thermal equilibrium (and generate quantum signatures via incompatibility) and check how long thermalisation \CYtwo{removes} quantum signatures (Fig.~\ref{Fig:Framework}).

\section{Results}
\subsection{Instruments and their incompatibility}
% {\bf\em Instruments and their incompatibility.---}
We start with briefly reviewing key ingredients. \RFtwo{To begin with,} {\em channels} are completely-positive trace-preserving linear maps~\cite{QIC-book}. They describe deterministic dynamics of quantum states. When dynamics become stochastic, they are described by {\em filters}, i.e., completely-positive trace-non-increasing linear maps. A filter $\mathcal{E}$ describes the following process: With an input state $\rho$, it outputs the state $\mathcal{E}(\rho)/{\rm tr}[\mathcal{E}(\rho)]$ with success probability ${\rm tr}[\mathcal{E}(\rho)]$. Physically, a filter is part of a channel~\footnote{For any filter $\mathcal{E}$, there exists another filter $\mathcal{E}'$ such that $\mathcal{E}+\mathcal{E}'$ is a channel.}. Importantly, filters can describe physical measurements via an {\em instrument}. Formally, an instrument is a set of filters, $\{\mathcal{E}_a\}_a$, such that $\sum_a\mathcal{E}_a$ is a channel. $\sum_a\mathcal{E}_a$ is called the {\em average channel} of this instrument. It models a general measurement process: Each index $a$ denotes a measurement outcome. For an input state $\rho$, ${\rm tr}[\mathcal{E}_a(\rho)]$ is the probability of getting outcome $a$ with the post-measurement state $\mathcal{E}_a(\rho)/{\rm tr}[\mathcal{E}_a(\rho)]$. This generalises positive operator-valued measures~\cite{QIC-book} (which only addresses measurement statistics).

Now, suppose we have several measurements, indexed by $x$, each is described by an instrument $\{\mathcal{E}_{a|x}\}_a$. Each $a$ represents a measurement outcome; i.e., \CYtwo{``$a|x$''} means that it is the $a$-th outcome of the $x$-th instrument. Collectively, we write $\mathbfcal{E}\coloneqq\{\mathcal{E}_{a|x}\}_{a,x}$ \CYtwo{as} a collection of instruments. A fundamental question is: {\em Can these instruments be jointly implemented?} Namely, \CYtwo{we ask whether there is} a single instrument $\{\mathcal{G}_\lambda\}_{\lambda}$ (\CYtwo{with measurement outcomes denoted} by $\lambda$'s) and \CYtwo{conditional probabilities} $\{P(a|x,\lambda)\}_{a,x,\lambda}$~\footnote{\CYtwo{$\sum_aP(a|x,\lambda)=1$ $\forall\,x,\lambda$; $P(a|x,\lambda)\ge0$ $\forall\,a,x,\lambda$.}} \CYtwo{achieving}
\begin{align}\label{Eq:instrument-JM}
\RFone{\mathcal{E}_{a|x}=\sum_\lambda P(a|x,\lambda) \mathcal{G}_{\lambda}\quad\forall\;a,x.}
\end{align}
If so, we say $\mathbfcal{E}$ is {\em compatible}~\cite{Mitra2022,Mitra2023,Buscemi2023,Heinosaari2016,Heinosaari2014,Ji2023}. Physically, \CYtwo{this means} a single device $\{\mathcal{G}_\lambda\}_\lambda$ plus classical post-measurement processing can reproduce \CYtwo{all $\mathbfcal{E}$'s outcomes}. When this expression is impossible, $\mathbfcal{E}$ \CYtwo{is called} {\em incompatible}. \RFtwo{A compatible $\mathbfcal{E}$ must satisfy $\sum_a\mathcal{E}_{a|x} = \sum_\lambda\mathcal{G}_\lambda$ $\forall\,x$---all its instruments have the same average channel $\sum_\lambda\mathcal{G}_\lambda$}. Hence, if $\mathbfcal{E}$ assigns different average channels to different $x$ values, $\mathbfcal{E}$ is trivially incompatible. To avoid \CYtwo{triviality}, we always assume \RFtwo{$\sum_a\mathcal{E}_{a|x} = \sum_b\mathcal{E}_{b|y}$ $\forall\,x,y$;} \RFone{i.e., all $\mathbfcal{E}$'s instruments have the same average channel}.

\subsection{Quantum steering in a temporal scenario}
% {\bf\em Quantum steering in a temporal scenario.---}
To detail our framework, one more ingredient is needed---{\em temporal steering scenario}~\cite{YNChen2014,CMLi2015,SLChen2016,Uola18} (see also Refs.~\cite{EPR,Schrodinger1935,Schrodinger1936,Wiseman2007PRL,Jones2007PRA,UolaRMP2020,Cavalcanti2016,Karthik2015JOSAB}). It is a multi-round scenario useful for certifying dynamical quantum resources~\cite{Ku2022PRXQ,HorodeckiRMP2003,RossetPRX2018,Vieira2024,Yuan2021npjQI,Abiuso2024,NarasimhacharPRL2019}. To start with, let $\mathbfcal{E}$ be a given collection of instruments. In each round, a state $\rho$ is prepared. With equal probability, an agent $A$ randomly picks an index $x$ and then measures the system via the $x$-th instrument in $\mathbfcal{E}$. Suppose the outcome is $a$. $A$ then sends the classical indices $(a,x)$ to another agent $B$. Also, $A$ uses a channel $\mathcal{N}$ to send the post-measurement state to $B$. After several rounds, $B$ obtains a collection of un-normalised states, ${\bm\sigma}\coloneqq\{\sigma_{a|x}\coloneqq\mathcal{N}\left[\mathcal{E}_{a|x}(\rho)\right]\}_{a,x}$, called a {\em state assemblage}. Physically, conditioned on $x$, ${\rm tr}(\sigma_{a|x})$ is the probability of obtaining outcome $a$ and the state $\sigma_{a|x}/{\rm tr}(\sigma_{a|x})$. \RFtwo{Notably, $\sum_a\sigma_{a|x} = \sum_b\sigma_{b|y}$ $\forall\,x,y$ as all $\mathbfcal{E}$'s instruments have the same average channel}. $\sum_a\sigma_{a|x}$ is called the {\em reduced state} of ${\bm\sigma}$.

% A 
\CYthree{Temporal steering scenarios 
% aims to 
can certify whether
% certify that 
$A$ 
% can 
influences $B$'s states via quantum measurement processes}. This can be done by checking if ${\bm\sigma}$ admits a description involving {\em no measurement}~\cite{YNChen2014,CMLi2015,SLChen2016}, 
% which is 
\CYthree{conventionally} called a {\em local hidden-state} (LHS) model. Mathematically, it is written as 
\begin{align}\label{Eq:LHS-definition}
\RFone{\sigma_{a|x}\stackrel{\rm LHS}{=}\sum_\lambda P_\lambda P(a|x,\lambda) \rho_\lambda\quad\forall\,a,x}
\end{align}
for some hidden variable $\lambda$, (conditional) probability distributions $\{P_\lambda\}_\lambda,\{P(a|x,\lambda)\}_{a,x,\lambda}$ and states $\rho_\lambda$. Physically, it tells us that although ${\bm\sigma}$ is produced by performing different measurements ($x$) on the input state, it can be simulated by doing {\em no} measurement---agent $A$ prepares the state $\rho_{\lambda}$ with probability $P_\lambda$, and then process the classical indices via $P(a|x,\lambda)$. Let ${\bf LHS}$ be the set of all LHS state assemblages. In quantum theory, significantly, there exists ${\bm\sigma}\notin{\bf LHS}$. Such a state assemblage is called {\em steerable}, as $A$ can influence (i.e., ``steer'') $B$'s states quantum-mechanically via some measurements.

% \section{Results}
\subsection{Framework}
% {\bf\em Framework.---}
Now, we detail our framework. For a given collection of instruments $\mathbfcal{E}$, we aim to thermodynamically quantify its incompatibility. To this end, we first measure a thermal equilibrium state via $\mathbfcal{E}$. This will drive the system out of thermal equilibrium, generating quantum signatures from $\mathbfcal{E}$'s incompatibility. Then, we let the system thermalise, and we analyse the minimal time needed to erase $\mathbfcal{E}$'s incompatibility signature (Fig.~\ref{Fig:Framework}). Conceptually, our framework is a temporal steering scenario where the input state $\rho$ is in thermal equilibrium, and the channel $\mathcal{N}$ is thermalisation.

Formally, our framework starts with a finite-dimensional quantum system in thermal equilibrium described by a {\em thermal state}: $\gamma = e^{-\frac{H}{k_BT}}/{\rm tr}\left(e^{-\frac{H}{k_BT}}\right),$ where $k_B$ is the Boltzmann constant. This describes the Boltzmann distribution with system Hamiltonian $H$ and background temperature $T$. In this work, we always assume $T$ is fixed, finite, and strictly positive, and we only consider Hamiltonians with finite energies. Hence, $\gamma$ is always full-rank. We then apply instruments $\mathbfcal{E}$ on this system just like in temporal steering scenarios: We randomly and uniformly pick an index $x$ and then measure the thermal equilibrium state $\gamma$ via the $x$-th instrument in $\mathbfcal{E}$. When the outcome is $a$, we obtain the un-normalised state $\mathcal{E}_{a|x}(\gamma)$. Collectively, we write $\mathbfcal{E}(\gamma)\coloneqq\{\mathcal{E}_{a|x}(\gamma)\}_{a,x}$, which is a state assemblage. Then, we let the system thermalise via a model
\begin{align}\label{Eq:ThermalisationModel}
\RFone{\mathcal{D}_t^{(h)}(\cdot)\coloneqq h(t)\mathcal{I}(\cdot) + \left[1-h(t)\right]\gamma{\rm tr}(\cdot),}
\end{align}
where $\mathcal{I}(\cdot)$ is the identity channel and \mbox{$h(t):[0,\infty)\to[0,1]$} is a strictly decreasing continuous function satisfying \mbox{$h(0)=1$} and \mbox{$\lim_{t\to\infty}h(t)=0$}. \RFone{Here, the superscript ``$(h)$'' means $\mathcal{D}_t^{(h)}$ is induced by the function $h(t)$ (abbreviated as $h$).} Physically, $\mathcal{D}_t^{(h)}$ models how a system \CYthree{reaches} thermal equilibrium when $t\to\infty$ \RFtwo{under evolution described by $h$. The function $h$ plays an important thermodynamic role, as different thermodynamic conditions lead to different functions $h$. E.g., in a qubit, by imposing preservation of thermal state and quantum detailed balance condition~\cite{Agarwal1973}, we obtain a thermalisation model $\mathcal{D}_t^{(h)}=\mathcal{D}_t^{(h_{\rm pt})}$ with $h(t) = h_{\rm pt}(t)\coloneqq e^{-t/t_0}$~\cite{Roga2010} called {\em partial thermalistaion} (abbreviated as ``$h_{\rm pt}$''), where $t_0$ is the thermalisation time scale (detailed in Supplemental Material I). This is a simple thermalisation model~\cite{Hsieh2020PRR,Hsieh2021PRXQ,Hsieh2020} also obtainable from collision models~\cite{Scarani2002}. Different $h$ functions represent different thermodynamic conditions (see also Refs.~\cite{De_SantisPRA2021,De_Santis2023Quantum}). We keep our approach fully general (Fig.~\ref{Fig:Framework}), so our results apply to a broad range of setups.}

\begin{figure}%[H]
\begin{center}
\scalebox{0.8}{\includegraphics{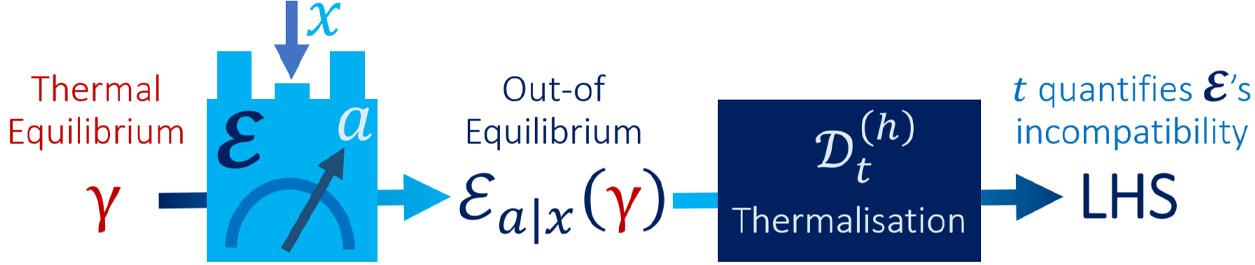}}
\caption{\CYtwo{\bf Framework.}}
\label{Fig:Framework} 
\end{center}
\end{figure}

\RFone{To certify incompatibility, note that $\mathbfcal{E}(\gamma)\in{\bf LHS}$ if $\mathbfcal{E}$ is compatible~\cite{Karthik2015JOSAB}---since} a compatible $\mathbfcal{E}$ achieves $\mathcal{E}_{a|x}(\gamma) = \sum_\lambda P_\lambda P(a|x,\lambda) \rho_\lambda$ with $P_\lambda={\rm tr}[\mathcal{G}_\lambda(\gamma)]$, $\rho_\lambda=\mathcal{G}_{\lambda}(\gamma)/{\rm tr}[\mathcal{G}_\lambda(\gamma)]$ (\CYtwo{$\{\mathcal{G}_\lambda\}_\lambda$'s outcomes act as} hidden variables). Hence, detecting steering from $\mathbfcal{E}(\gamma)$ \CYtwo{certifies $\mathbfcal{E}$'s incompatibility}. We thus consider the following figure-of-merit, which is the longest time \CYtwo{that} $\mathbfcal{E}(\gamma)$ \CYtwo{is} steerable under thermalisation:
\begin{align}\label{Eq:thermalisation time}
t_{\rm min}^{(h)}(\mathbfcal{E})\coloneqq\min\left\{t\ge0\,\middle|\,\left\{\mathcal{D}_t^{(h)}\left[\mathcal{E}_{a|x}(\gamma)\right]\right\}_{a,x}\in{\bf LHS}\right\}.
\end{align}
\RFtwo{Note that $t_{\rm min}^{(h)}$ is $\gamma$-dependent~\footnote{\RFtwo{Different $\gamma$ gives different thermalisation model via Eq.~\eqref{Eq:ThermalisationModel}.}}}, \RFone{and the superscript ``$(h)$'' denotes the $h$-dependence.} If $t_{\rm min}^{(h)}(\mathbfcal{E})>0$, \CYtwo{the collection} $\{\mathcal{D}_t^{(h)}\circ\mathcal{E}_{a|x}\}_{a,x}$ is incompatible during the time window $0\le t<t_{\rm min}^{(h)}(\mathbfcal{E})$, which thus certifies $\mathbfcal{E}$'s incompatibility. Moreover, a higher $t_{\rm min}^{(h)}(\mathbfcal{E})$ implies that $\mathbfcal{E}$'s incompatibility signature can survive longer during thermalisation.

\subsection{Incompatibility signature can survive thermalisation}
% {\bf\em Incompatibility signature can survive thermalisation.---}
\CYtwo{Now,} there are two natural questions. First, {\em can we really have $t_{\rm min}^{(h)}(\mathbfcal{E})>0$?} Second, {\em how to compute \CYtwo{$t_{\rm min}^{(h)}$}?} To answer them, \CYtwo{we show that} $t_{\rm min}^{(h)}$ \CYtwo{equals} a novel steering quantifier that is computable via {\em semi-definite programming} (SDP)~\cite{watrous_2018,SDP-textbook}. To see this, we introduce the {\em (logarithmic) thermalisation steering robustness} subject to $\gamma$, denoted by ${\rm SR}_\gamma$. \CYtwo{For} a state assemblage ${\bm\sigma} = \{\sigma_{a|x}\}_{a,x}$, we define $2^{{\rm SR}_\gamma({\bm\sigma})}$ as
\begin{align}\label{Eq:SRgamma}
\min\left\{\RFone{\mu}\ge1\,\middle|\,\left\{\frac{\sigma_{a|x} + (\RFone{\mu}-1){\rm tr}(\sigma_{a|x})\gamma}{\RFone{\mu}}\right\}_{a,x}\in{\bf LHS}\right\}.
\end{align}
It quantifies steering as ${\rm SR}_\gamma({\bm\sigma})=0$ if and only if ${\bm\sigma}\in{\bf LHS}$, and it is non-increasing under certain allowed operations of steering~\cite{Gallego2015PRX}. \RFone{In Supplemental Material II, we show that
\begin{result}\label{Eq:obs1}
$h\left(t_{\rm min}^{(h)}(\mathbfcal{E})\right) = 2^{-{\rm SR}_\gamma[\mathbfcal{E}(\gamma)]}$ for every $\mathbfcal{E},\gamma$.
\end{result}
Hence, the informational-theoretic quantity} ${\rm SR}_\gamma$ carries the physical meaning as the longest time for incompatibility to survive thermalisation \RFone{(see Appendix A for examples)}. \CYtwo{Since ${\rm SR}_\gamma({\bm\sigma})>0$ for ${\bm\sigma}\notin{\bf LHS}$, we conclude} that $t_{\rm min}^{(h)}>0$ is indeed possible---as long as \mbox{$\mathbfcal{E}(\gamma)\notin{\bf LHS}$,} incompatibility signature can be seen within a {\em non-vanishing} time window $0\le t < t_{\rm min}^{(h)}(\mathbfcal{E})$. This thus answers the first question. \RFtwo{In fact, incompatibility can {\em always} be certified by $t_{\rm min}^{(h)}$:
\begin{result}\label{result: environment result}
% For an  
\CYthree{$\mathbfcal{E}$ is incompatible (in system $S$) if and only if} there exists an auxiliary system $E$ (with the same dimension as $S$) and a full-rank $\gamma$ in $SE$ such that $\{(\mathcal{E}_{a|x}\otimes\mathcal{I})(\gamma)\}_{a,x}\notin{\bf LHS}$.
\end{result}
% Together 
\CYthree{With} Result~\ref{Eq:obs1}, incompatibility signature can always be seen in our framework \CYthree{(see Appendix A for remarks)}. Moreover, a one-parameter family of full-rank states $\gamma$'s can universally certify all incompatible instruments. See Supplemental Material III for proofs.} \CYtwo{To answer the second question,} note that Eq.~\eqref{Eq:SRgamma} is an SDP, which is \CYtwo{numerically feasible and efficiently computable} (see \CYtwo{Supplemental Material IV}). Interestingly, although the thermalisation time $t_{\rm min}^{(h)}$ could be hard to measure in practice, we can use ${\rm SR}_\gamma$ to compute it numerically. ${\rm SR}_\gamma$ is thus a numerical tool for estimating the longest time for incompatibility to survive different thermalisation models.

\subsection{Resource theory of incompatible instruments with thermodynamic constraints}
% {\bf\em Resource theory of incompatible instruments with thermodynamic constraints.---}
Actually, $t_{\rm min}^{(h)}$ can {\em quantify} incompatible instruments \CYtwo{via a} \RFone{novel resource theory of incompatibility, as we detail below}. A {\em resource theory}~\cite{ChitambarRMP2019} is defined by the sets of free objects (i.e., objects without the resource) and allowed operations (i.e., operations allowed to manipulate the resource). In our resource theory, the resource is incompatibility, and free objects are compatible instruments. Now, we detail the operations allowed to manipulate a given $\mathbfcal{E} = \{\mathcal{E}_{b|y}\}_{b,y}$. \RFone{We will discuss deterministic and stochastic allowed operations separately. For the deterministic ones,} the first thing we allow is pre- and post-processing classical indices. Namely, we still randomly pick some $x$ value. Then, we use a classical distribution $P(y|x)$ to produce an index $y$, which is the input for the instruments $\mathbfcal{E} = \{\mathcal{E}_{b|y}\}_{b,y}$. Suppose its measurement outcome is $b$. Then, conditioned on the knowledge of $x,y,b$, we use another classical distribution $P'(a|xyb)$ to produce the outcome $a$. By summing over $b,y$, we get new instruments \RFone{\mbox{$\{\mathcal{L}_{a|x}\}_{a,x} = \{\sum_{b,y}P'(a|xyb)P(y|x)\mathcal{E}_{b|y}\}_{a,x}$}} solely from classical processing. Apart from this, we further allow adding two {\em Gibbs-preserving channels}~\cite{Lostaglio_2019,Faist2015NJP,HorodeckiPRL2003,HorodeckiPRA2003}, $\mathcal{P}$ and \CYtwo{$\mathcal{Q}$}, before and after the above process (i.e., $\mathcal{L}_{a|x}\mapsto\CYtwo{\mathcal{Q}}\circ\mathcal{L}_{a|x}\circ\mathcal{P}$ \RFone{$\forall\,a,x$}). They satisfy $\CYtwo{\mathcal{Q}}(\gamma) = \gamma = \mathcal{P}(\gamma)$, i.e., they cannot drive thermal equilibrium out of equilibrium. Altogether, these define the {\em deterministic allowed operations}, \RFone{denoted by $\mathbb{F}$,} in our resource theory: \RFone{$\mathbfcal{E}\mapsto\mathbb{F}(\mathbfcal{E})\coloneqq\{\mathbb{F}(\mathbfcal{E})_{a|x}\}_{a,x}$, where $\mathbb{F}(\mathbfcal{E})$ is a collection of instruments whose elements are defined by
\begin{align}\label{Eq:DAO}
\mathbb{F}(\mathbfcal{E})_{a|x}\coloneqq\sum_{b,y}P'(a|xyb)P(y|x)\mathcal{Q}\circ\mathcal{E}_{b|y}\circ\mathcal{P}\quad\forall\,a,x.
\end{align}
We thus obtain a novel resource theory of instrument incompatibility, which is the ones in Refs.~\cite{Heinosaari2014,Buscemi2023,Ji2023} subject to Gibbs-preserving condition. Notably, Eq.~\eqref{Eq:DAO} is fully general. More thermodynamic conditions can be imposed if needed. See Appendix B for remarks and schematic illustration (Fig.~\ref{Fig:DAO}).} Importantly, deterministic allowed operations cannot generate incompatibility from \CYtwo{compatible ones---}$\mathbb{F}(\mathbfcal{E})$ is compatible if $\mathbfcal{E}$ is compatible. Hence, their output is less resourceful than the input. If $ t_{\rm min}^{(h)}$ is a suitable measure of incompatibility, it should also obey this rule---it cannot increase under deterministic allowed operations. As proved in Supplemental Material V, this is indeed the case:
\begin{result}\label{Result:ResourceMonotone}
\mbox{$t_{\rm min}^{(h)}[\mathbb{F}(\mathbfcal{E})]\le t_{\rm min}^{(h)}(\mathbfcal{E})$} for every deterministic allowed operation $\mathbb{F}$, and \mbox{$t_{\rm min}^{(h)}(\mathbfcal{E})=0$} if $\mathbfcal{E}$ is compatible.
\end{result} 
Hence, $t_{\rm min}^{(h)}$ is \CYtwo{an appropriate quantifier} subject to the deterministic allowed operations.

\RFone{Now, let us consider stochastic allowed operations, which} are important in stochastic manipulations of quantum resources~\cite{Regula2022Quantum,Regula2022PRL,Regula2024reversibility,Ku2022NC,Ku2023,Hsieh2023}. Here, we consider a simple type of stochastic operations called {\em ${\rm LF_1}$ filters}~\cite{Ku2022NC,Ku2023,Hsieh2023}. Formally, a ${\rm LF_1}$ filter takes the form $(\cdot)\mapsto K(\cdot)K^\dagger$ with $K^\dagger K\le\id$; i.e., it is a filter with a single Kraus operator. It is simple enough to be practically feasible, and non-trivial enough to, e.g., optimally distil steering~\cite{Ku2022NC,Ku2023,Hsieh2023}. Here, we investigate whether ${\rm LF_1}$ filters can increase quantum signature's survival time under thermodynamic constraints. Suppose we have produced $\mathbfcal{E}(\gamma)$. Then we ask: {\em Can ${\rm LF_1}$ filters increase $t_{\rm min}^{(h)}$ without changing measurement statistics and thermal state?} Namely, we consider ${\rm LF_1}$ filters satisfying (i) ${\rm tr}[K\mathcal{E}_{a|x}(\gamma) K^\dagger] / p_\gamma = {\rm tr}[\mathcal{E}_{a|x}(\gamma)]$ $\forall\,a,x$, and (ii) $K\gamma K^\dagger / p_\gamma=\gamma$~\footnote{\CYthree{As an example of a filter with condition (ii), consider $K=\sqrt{p}U$ with $0<p\le1$ and an energy-conserving unitary $U$.}}. Here, $p_\gamma = {\rm tr}(K\gamma K^\dagger)$ is filter's success probability. Physically, condition (i) means that the system looks the same if we only \CYtwo{examine} classical statistics. Condition (ii) \CYtwo{means the} filter cannot drive thermal equilibrium out of equilibrium. Our {\em stochastic allowed operations} are ${\rm LF_1}$ filters with conditions (i) and (ii) \CYthree{(they are closed under composition; see Supplemental Material VI)}.
\CYtwo{As proved in Supplemental Material VI,} we have:
\begin{result}\label{Result:No-extension}
Suppose $\mathbfcal{E}$'s average channel is Gibbs-preserving.
Then no stochastic allowed operation can increase $t_{\rm min}^{(h)}(\mathbfcal{E})$.
\end{result}
Hence, to increase $t_{\rm min}^{(h)}$ stochastically, we need to either allow $\mathbfcal{E}$'s average channel to \CYtwo{generate non-equilibrium from equilibrium}, or drop one of the conditions (i) and (ii). In particular, if $\mathbfcal{E}$'s average channel is Gibbs-preserving, the ${\rm LF_1}$ filter needs to either drive thermal equilibrium out of equilibrium, or change the classical statistics of the measurements. \CYtwo{Result~\ref{Result:No-extension}} thus uncovers the cost of stochastically increasing $t_{\rm min}^{(h)}$. Notably, $t_{\rm min}^{(h)}$ is {\em monotonic} under stochastic allowed operations\CYtwo{---it is} a quantifier even with stochastic manipulations. Interestingly, \CYtwo{Results~\ref{Eq:obs1} and~\ref{Result:No-extension} imply} that we {\em cannot} stochastically distil steering in the current thermodynamic setting. Since ${\rm LF_1}$ filters are powerful for steering distillation~\cite{Ku2022NC,Ku2023,Hsieh2023}, Result~\ref{Result:No-extension} shows that imposing thermodynamic constraints can strongly limit the strength of distilling quantum resources.

\subsection{Incompatible advantages in work extraction}
% {\bf\em Incompatible advantages in work extraction.---}
After knowing how to use thermalisation to quantify incompatibility, a natural question is whether any other thermodynamic task can also do so. An appropriate answer will be useful for studying quantum signatures in different thermodynamic contexts~\cite{Oppenheim2002,delRio2011,Lostaglio2020PRL,Levy2020PRXQ,Puliyil2022PRL,Upadhyaya2023,JiPRL2022,BeyerPRL2019,ChanPRA2022,Centrone2024,Hsieh2023IP,Lipka-BartosikPRL2024,NavascuesPRL2015}. To this end, we utilise the work extraction task \CYtwo{introduced} in Ref.~\cite{Hsieh2023IP}. Consider a finite-dimensional system with Hamiltonian $H$ and \CYthree{a background} temperature $T$. When this system is in a state $\rho$, one can (on average) extract the following optimal amount of work from it (e.g., by the work extraction scenario in Ref.~\cite{Skrzypczyk2014}): $W_{{\rm ext}}(\rho,H) = (k_BT\ln2)D\left(\rho\,\|\,\gamma_H\right),$ where $D(\rho\,\|\,\sigma)\coloneqq{\rm tr}\left[\rho\left(\log_2\rho - \log_2\sigma\right)\right]$ is the {\em quantum relative entropy}. Here, $\gamma_H$ is the thermal state associated with $H$ and $T$ \CYtwo{(we make the Hamiltonian dependence explicit)}. When \CYtwo{$H=0$}, we have $W_{{\rm inf}}(\rho) \coloneqq W_{{\rm ext}}(\rho,0) = (k_BT\ln2)D\left(\rho\,\|\,\id/d\right).$ This is the optimal work extractable from $\rho$'s information content. Then, Ref.~\cite{Hsieh2023IP} considers the deficit-type work value (inspired by Ref.~\cite{Oppenheim2002}): $\Wdiff(\rho,H)\coloneqq W_{{\rm ext}}(\rho,H) - W_{{\rm inf}}(\rho).$ It characterises the difference between $\rho$'s extractable work in two scenarios: One \CYtwo{with $H$; one with} zero Hamiltonian \RFone{(see also Fig.~1 in Ref.~\cite{Hsieh2023IP})}. Physically, $\Wdiff$ is the energy extracted from $H$'s \RFtwo{\em non-vanishing energy difference} when the system is in $\rho$. As proved in Ref.~\cite{Hsieh2023IP}, the work deficit $\Wdiff$ can certify 
% {\em general} quantum resources, 
\CYthree{quantum resources such as entanglement, coherence, etc}. We now show that it can actually {\em quantify} steering and incompatibility.

\RFone{Consider four batches of experiments (Fig.~\ref{Fig:Wdiff}), with a given thermal state $\gamma_{\rm H_{\rm in}}$ (subject to Hamiltonian $H_{\rm in}$), a fixed set of Hamiltonians ${\bf H}=\{H_{a|x}\}_{a,x}$, and $\mathbfcal{E}$ whose average channel preserves $\gamma_{H_{\rm in}}$. In the first batch, we apply $\mathbfcal{E}$ on $\gamma_{H_{\rm in}}$, producing} ${\bm\sigma}=\mathbfcal{E}(\gamma_{H_{\rm in}})$. With uniformly distributed \mbox{$x=0,...,|{\bf x}|-1$}, we obtain the state \mbox{$\hat{\sigma}_{a|x}\coloneqq\sigma_{a|x}/{\rm tr}(\sigma_{a|x})$} with probability $P(a,x)\coloneqq{\rm tr}(\sigma_{a|x})/|{\bf x}|$. \RFone{Upon obtaining indices $(a,x)$, we quench the Hamiltonian to $H_{a|x}$ and extract work from $\hat{\sigma}_{a|x}$, obtaining $W_{{\rm ext}}(\hat{\sigma}_{a|x},H_{a|x})$. Sufficiently many trials give $\sum_{a,x}P(a,x)W_{{\rm ext}}(\hat{\sigma}_{a|x},H_{a|x})$. The second batch is identical to the first, except we always quench the Hamiltonian to zero rather than $H_{a|x}$. This gives $\sum_{a,x}P(a,x)W_{{\rm inf}}(\hat{\sigma}_{a|x})$. Using these two batches, we obtain $\sum_{a,x}P(a,x)\Wdiff(\hat{\sigma}_{a|x},H_{a|x})$. In the third batch, we start with $\gamma_{H_{\rm in}}$, quench its Hamiltonian to $H_{a|x}$ with probability $P(a,x)$ (obtained in the first two batches), and extract work, obtaining $\sum_{a,x}P(a,x)W_{{\rm ext}}(\gamma_{H_{\rm in}},H_{a|x})$. The fourth batch is identical to the third, except we always quench the Hamiltonian to zero, resulting in $\sum_{a,x}P(a,x)W_{{\rm inf}}(\gamma_{H_{\rm in}})$. These two batches give $\sum_{a,x}P(a,x)\Wdiff(\gamma_{H_{\rm in}},H_{a|x})$. See also Fig.~\ref{Fig:Wdiff}. Let $\gamma = \gamma_{H_{\rm in}}$. Then, we consider the following figure-of-merit} $\WdiffSA({\bm\sigma},{\bf H})~\coloneqq~\sum_{a,x}P(a,x)\left[\Wdiff(\hat{\sigma}_{a|x},H_{a|x}) - \Wdiff(\gamma,H_{a|x})\right],$ \CYtwo{which} is a form of {\em energy}. Hence, if we can use $\WdiffSA$ to express ${\rm SR}_\gamma$, then {\em ${\rm SR}_\gamma$ can be evaluated by energy extraction experiments.} This can generalise the recent findings on certifying quantum steering via heat engines~\cite{JiPRL2022,BeyerPRL2019,ChanPRA2022} to instrument incompatibility. In \RFtwo{Supplemental Material VIII, we show this is possible. For a parameter $0<\eta<\infty$, let $\mathbfcal{H}_\eta({\bm\sigma})\coloneqq\{{\bf H}\;|\; \WdiffSA({\bm\sigma},{\bf H})\ge k_BT\eta\}$ containing ${\bf H}$'s with $\WdiffSA$ lower bounded by the energy scale $k_BT\eta$. Then we have:}

\begin{figure}%[H]
\begin{center}
\scalebox{0.8}{\includegraphics{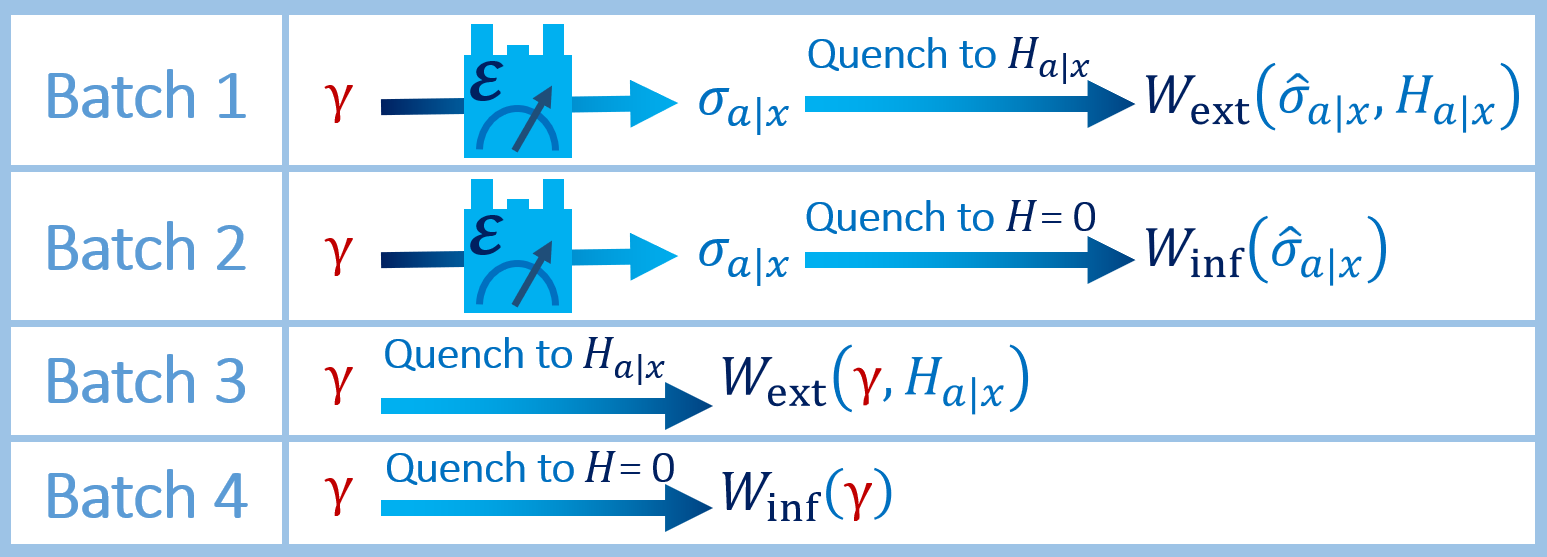}}
\caption{\RFone{{\bf Four-batches work extraction.}}}
\label{Fig:Wdiff} 
\end{center}
\end{figure}

\begin{result}\label{Result:WorkExtraction}
Let ${\bm\sigma}$ be a steerable state assemblage with reduced state $\gamma$ \CYtwo{and} ${\rm tr}(\sigma_{a|x})>0$ $\forall\,a,x$. Then, \RFtwo{for every $0<\eta<\infty$,} 
\begin{align}\label{Eq:Result:WorkExtraction}
2^{{\rm SR}_\gamma({\bm\sigma})} = \max_{{\bf H}\in\RFtwo{\mathbfcal{H}_\eta}({\bm\sigma})}\frac{\WdiffSA({\bm\sigma},{\bf H})}{\max_{{\bm\tau}\in{\bf LHS}({\bm\sigma})}\WdiffSA({\bm\tau},{\bf H})},
\end{align}
where ${\bf LHS}({\bm\sigma})\coloneqq\{{\bm\tau}\in{\bf LHS}\,|\,{\rm tr}(\sigma_{a|x})={\rm tr}(\tau_{a|x})\;\forall\,a,x\}$.
\end{result}
Hence, steering is quantified {\em necessarily and sufficiently} by the highest work extraction advantage over state assemblages in ${\bf LHS}({\bm\sigma})$. To see the implication to incompatibility, consider $\mathbfcal{E}$ with Gibbs-preserving average channel and set ${\bm\sigma}=\mathbfcal{E}(\gamma)$ (with reduce state $\gamma$). By \CYtwo{Results~\ref{Eq:obs1} and~\ref{Result:WorkExtraction}}, $t_{\rm min}^{(h)}(\mathbfcal{E})>0$ implies ${\rm SR}_\gamma({\bm\sigma})>0$, leading to work extraction advantage. \RFtwo{See Appendix C for further remarks and an illustrative example for certifying quantum advantage.}

\CYtwo{Unexpectedly, Result~\ref{Result:WorkExtraction} implies} that {\em we can use energy extraction experiments to obtain the time value $t_{\rm min}^{(h)}$.} To see this, suppose ${\bf H_*}$ are Hamiltonians achieving the maximisation in Eq.~\eqref{Eq:Result:WorkExtraction} for ${\bm\sigma}=\mathbfcal{E}(\gamma)$. With partial thermalisation $h_{\rm pt}(t)\coloneqq e^{-t/t_0}$ and \CYtwo{Result~\ref{Eq:obs1} [i.e., Eq.~\eqref{Eq:example-pt}]}, we obtain $t_{\rm min}^{\left(h_{\rm pt}\right)}(\mathbfcal{E})/t_0 = \ln\left(\WdiffSA\left[\mathbfcal{E}(\gamma),{\bf H}_*\right]/\max_{{\bm\tau}\in{\bf LHS}({\bm\sigma})}\WdiffSA({\bm\tau},{\bf H}_*)\right).$ \CYtwo{For any $\varepsilon\ge1$}, detecting $\WdiffSA\left[\mathbfcal{E}(\gamma),{\bf H}_*\right]/\max_{{\bm\tau}\in{\bf LHS}({\bm\sigma})}\WdiffSA({\bm\tau},{\bf H}_*)~>~\CYtwo{\varepsilon}$ \CYtwo{is thus {\em equivalent to}} $t_{\rm min}^{\left(h_{\rm pt}\right)}(\mathbfcal{E}) > t_0\ln\CYtwo{\varepsilon}$. This illustrates how energy extraction, which is completely different from time measurement, \CYtwo{can estimate} the thermalisation time $t_{\rm min}^{\left(h\right)}$. Finally, in \CYtwo{Supplemental Material VIII}, we show that Hamiltonians certifying quantum signatures can be efficiently found via SDP.

\subsection{Quantum signature vanishes when system thermalises}
\RFtwo{
% {\bf\em Quantum signature vanishes when system thermalises.---}
So far, we only discuss thermalisation via {\em Markovian} evolutions. If a {\em general} evolution thermalises a system, will quantum signature vanish at finite time? Here, an {\em evolution} is a one-parameter family of {\em positive} trace-preserving linear maps $\{\mathcal{N}_t\}_{t=0}^\infty$. It describes how a state $\rho$ evolves to $\mathcal{N}_t(\rho)$ at time $t$. We say $\{\mathcal{N}_t\}_{t=0}^\infty$ {\em thermalises} the system to thermal state $\gamma$ if $\lim_{t\to\infty}\norm{\mathcal{N}_t(\cdot) - \gamma{\rm tr}(\cdot)}_\diamond = 0$ ($\norm{\cdot}_\diamond$ is the {\em diamond norm}~\cite{Aharonov1998,watrous_2018}). Can non-Markovianity {\em protect} quantum signatures of steering and incompatibility? This is, surprisingly, {\em impossible} (proved in Supplemental Material VII):
\begin{result}\label{Result:Strong-no-go}
Let ${\bm\sigma}$ be a state assemblage with reduced state $\gamma$. Suppose ${\rm tr}(\sigma_{a|x})>0$ $\forall\,a,x$. If $\{\mathcal{N}_t\}_{t=0}^\infty$ thermalises the system to $\gamma$, there is a finite $t_*$ such that $\mathcal{N}_t({\bm\sigma})\in{\bf LHS}$ $\forall\;t>t_*.$
\end{result}
By considering $\mathbfcal{E}$ with a Gibbs-preserving average channel [and $\mathbfcal{E}(\gamma)$'s reduce state is $\gamma$] and set ${\bm\sigma} = \mathbfcal{E}(\gamma)$, we obtain the same no-go result for incompatibility signature. Hence, thermalisation can surpass the strength of non-Markovianity.}

\subsection{Discussions}
% {\bf\em Discussions.---}
Many questions remain open. First, making the work extraction game more experimentally feasible is important. This will be addressed in our follow-up projects. Second, it is valuable to study other types of incompatibility~\cite{Ku2018PRA,Ku2021PRR,Ku2022PRXQ,Hsieh2024arXiv,Mitra2022,Mitra2023,Buscemi2023,Heinosaari2016,Heinosaari2014,Ji2023,Hsieh2022PRR,Haapasalo2021,Hsieh2023-2} in a thermodynamic context. Finally, it is interesting to know whether ${\rm SR}_\gamma$ is (sub-)additive, which is subtle due to the superactivation of steering~\cite{Hsieh2016PRA,Quintino2016}.

\section{Acknowledgements}
% \begin{acknowledgements}
% {\bf\em Acknowledgement.---}
The authors acknowledge fruitful discussions with Manuel Gessner, Huan-Yu Ku, M\'at\'e Farkas, Bartosz Regula, Paul Skrzypczyk, Benjamin Stratton, and Hao-Cheng Weng. We especially thank Hao-Cheng Weng for comments and discussions on the experimental feasibility of the work extraction tasks in Nitrogen-Vacancy centre systems. 
C.-Y.~H. acknowledges support from the Royal Society through Enhanced Research Expenses (on grant NFQI), the ERC Advanced Grant (on grant FLQuant), and the Leverhulme Trust Early Career Fellowship (on grant ``Quantum complementarity: a novel resource for quantum science and technologies'' with grant number ECF-2024-310). 
S.-L.~C. is supported by the National Science and Technology Council (NSTC) Taiwan (Grant No.~NSTC 111-2112-M-005-007-MY4) and National Center for Theoretical Sciences Taiwan (Grant No.~NSTC 113-2124-M-002-003).
% \end{acknowledgements}

\section{Appendix}
\subsection{Appendix A: Illustrative example of Result~\ref{Eq:obs1} and remarks on Result~\ref{result: environment result}}
\RFone{
% {\bf\em Appendix A: Illustrative example of Result~\ref{Eq:obs1} \CYthree{and remarks on Result~\ref{result: environment result}}.---}
To illustrate the physics of Result~\ref{Eq:obs1}, let us use the partial thermalisation $h(t) = h_{\rm pt}(t)\coloneqq e^{-t/t_0}$. 
For simplicity, let us write $t_* = t_{\rm min}^{\left(h_{\rm pt}\right)}(\mathbfcal{E})$.
Recall that $h_{\rm pt}$ is the abbreviation of the function $h_{\rm pt}(t)$, and the superscript ``$(h_{\rm pt})$'' denotes the $h_{\rm pt}$-dependence.
Then we have, by definition, $h_{\rm pt}\left(t_*\right) = e^{-t_*/t_0}$.
On the other hand, setting $h=h_{\rm pt}$ in Result~\ref{Eq:obs1}, we obtain
$h_{\rm pt}\left(t_*\right) = 2^{-{\rm SR}_\gamma[\mathbfcal{E}(\gamma)]} = e^{-{\rm SR}_\gamma[\mathbfcal{E}(\gamma)]\ln2}$.
Taking the logarithm of both sides of the above equality \CYthree{gives}
% , we obtain
\begin{align}\label{Eq:example-pt}
% $
t_{\rm min}^{\left(h_{\rm pt}\right)}(\mathbfcal{E})/(t_0\ln2) = {\rm SR}_\gamma[\mathbfcal{E}(\gamma)].
% $
\end{align}
In this case, ${\rm SR}_\gamma$ is the thermalisation time in the unit $t_0\ln2$.

\CYthree{Now, we remark on Result~\ref{result: environment result}.
First, by changing $\mathbfcal{E}$ into \mbox{$\{\mathcal{E}_{a|x}\otimes\mathcal{I}\}_{a,x}$} in Eq.~\eqref{Eq:thermalisation time} and $\mathbfcal{E}(\gamma)$ into $\{(\mathcal{E}_{a|x}\otimes\mathcal{I})(\gamma)\}_{a,x}$ in Result~\ref{Eq:obs1}, Result~\ref{result: environment result} is connected to our framework (in $SE$).
Second, Result~\ref{result: environment result} needs an auxiliary system to certify all incompatible instruments.
% via thermalisation time [Eq.~\eqref{Eq:thermalisation time}].
An open problem is whether this is possible {\em without} utilising any auxiliary system, which is left for future projects.
Finally, the thermality of $\gamma$ in Result~\ref{result: environment result} is crucial if we want to combine Result~\ref{result: environment result}  with Eq.~\eqref{Eq:thermalisation time} and Result~\ref{Eq:obs1} to certify incompatibility via thermalisation time (which needs $\gamma$ to be the thermal state).
}
}

\subsection{Appendix B: Remarks on deterministic allowed operations and a schematic illustration}
% {\bf\em Appendix B: Remarks on deterministic allowed operations and a schematic illustration.---}
Many thermodynamically allowed operations existing in the literature are Gibbs-preserving plus additional thermodynamic conditions.
Equation~\eqref{Eq:DAO} thus provides a general way to include different additional thermodynamic conditions.
For instance, with the energy-conserving condition, $\mathcal{P},\mathcal{Q}$ in Eq.~\eqref{Eq:DAO} becomes thermal operations~\cite{Janzing2000,Brandao2013PRL} (see also, e.g., Refs.~\cite{Lostaglio_2019}), which are Gibbs-preserving channels that can be realised by an energy-conserving unitary and a thermal bath. 
Also, if we focus on incoherent dynamics, one can consider Gibbs-preserving incoherent channels $\mathcal{P},\mathcal{Q}$ in the given energy eigenbasis.
% Finally, 
\CYthree{See} Fig.~\ref{Fig:DAO} for a schematic illustration of Eq.~\eqref{Eq:DAO}.
\CYthree{Finally, as proved in Supplemental Material V,
% we remark that 
the composition of any two deterministic allowed operations is again a deterministic allowed operation---they are closed under composition.
% See Supplemental Material V for the proof.
}

\begin{figure}%[H]
\begin{center}
\scalebox{0.8}{\includegraphics{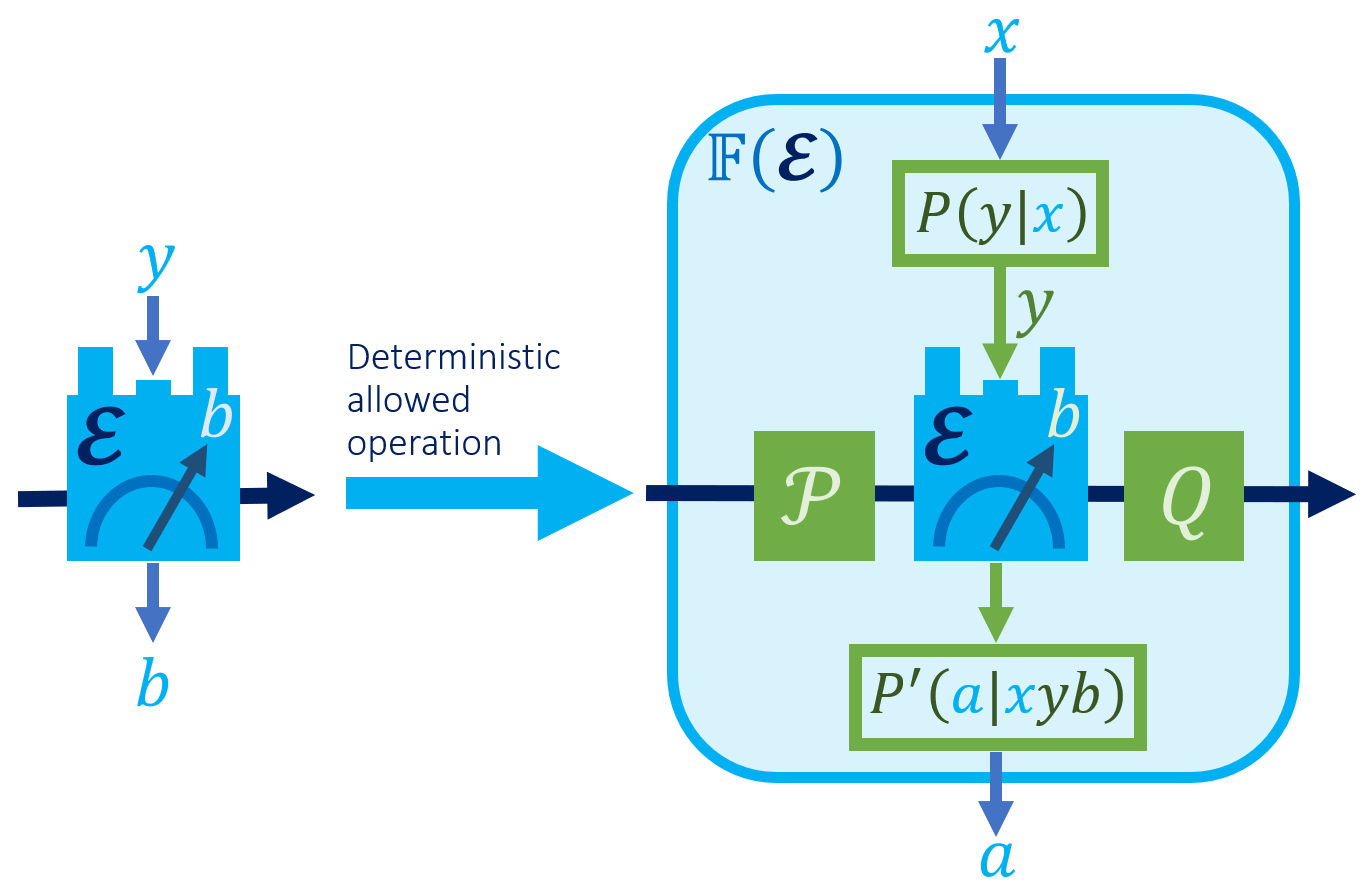}}
\caption{
\RFone{{\bf Illustration of deterministic allowed operations.}
Deterministic allowed operations provide a deterministic way to transform a collection of instruments ($\mathbfcal{E}$; box on the left-hand side) into another collection of instruments [$\mathbb{F}(\mathbfcal{E})$; light blue box on the right-hand side].
In the figure, horizontal lines are for quantum information processing; vertical lines are for classical data processing.
To achieve a deterministic allowed operation, one uses the classical pre-processing $P(y|x)$ to generate the index $y$ as the classical input of $\mathbfcal{E}$.
By obtaining measurement outcome (indicated by the classical index $b$), one further uses the post-processing $P'(a|xyb)$ to generate the index $a$.
Finally, one also uses a quantum pre- and post-processing channels $\mathcal{P},\mathcal{Q}$ to process the quantum system.
}
}
\label{Fig:DAO} 
\end{center}
\end{figure}

\subsection{Appendix C: Remarks and example for Result~\ref{Result:WorkExtraction}}
% \RFtwo{\bf\em Appendix C: Remarks and example for Result~\ref{Result:WorkExtraction}.---}
\RFtwo{
To certify quantum signatures, it {\em suffices} to check whether ${\bm\sigma}$ can outperform the set
${\bf LHS}({\bm\sigma})$. 
This is because ${\bf LHS}({\bm\sigma})$ contains all LHS state assemblages ${\bm\tau}$ whose classical statistics ${\rm tr}(\tau_{a|x})$'s are identical to ${\bm\sigma}$'s; i.e., ${\rm tr}(\tau_{a|x})={\rm tr}(\sigma_{a|x})$ $\forall\,a,x$.
Once we obtain ${\bm\sigma}$, we also know ${\rm tr}(\sigma_{a|x})$'s, and LHS state assemblages with different classical statistics surely cannot simulate ${\bm\sigma}$. 
Indeed, as shown in Ref.~\cite{Hsieh2023IP}, we have \mbox{${\bm\sigma}\notin{\bf LHS}$} if and only if ${\bm\sigma}\notin{\bf LHS}({\bm\sigma})$.
Hence, with the knowledge of classical statistics ${\rm tr}(\sigma_{a|x})$, quantum signatures can be claimed once ${\bf LHS}({\bm\sigma})$ is outperformed.
As ${\bf LHS}({\bm\sigma})$ is a convex set much smaller than ${\bf LHS}$, Result~\ref{Result:WorkExtraction} can certify quantum signature in a way more efficient than optimising over ${\bf LHS}$.
}

\CY{
Now, we provide an example to illustrate the experimental possibility of witnessing the quantum advantage in Result~\ref{Result:WorkExtraction}, i.e., \mbox{$2^{{\rm SR}_\gamma({\bm\sigma})}>1$}.
This is equivalent to demonstrating
\begin{align}\label{Eq:comp--1}
\max_{{\bm\tau}\in{\bf LHS}({\bm\sigma})}\WdiffSA({\bm\tau},{\bf H})< \WdiffSA({\bm\sigma},{\bf H})
\end{align} 
for some set of Hamiltonians ${\bf H}$ and some $\gamma$.
Since the work values $W_{\rm ext}, W_{\rm inf}$ can be obtained by existing work extraction scenarios (e.g., the one in Ref.~\cite{Skrzypczyk2014}), we focus on seeking experimentally feasible Hamiltonians and state assemblages.}

\CY{
The example starts with a qubit system and \mbox{$\gamma = \id/2$}.
Consider the Pauli $X,Z$ observables defined by \mbox{$X = \proj{+} - \proj{-}$} and \mbox{$Z = \proj{0} - \proj{1}$}, respectively, where \mbox{$\ket{\pm}\coloneqq(\ket{0}\pm\ket{1})/\sqrt{2}$}.
Then, by applying the corresponding projective measurements $\{\proj{+},\proj{-}\}$ and $\{\proj{0},\proj{1}\}$ on $\id/2$, we obtain a qubit state assemblage, ${\bm\sigma}^{\rm Pauli}\coloneqq\{\sigma_{a|x}^{\rm Pauli}\}_{a,x}$, given by
\begin{align}\label{Eq: sigma Pauli example}
&\sigma_{0|0}^{\rm Pauli} = \proj{+}/2,\;\sigma_{1|0}^{\rm Pauli} = \proj{-}/2,\\
&\sigma_{0|1}^{\rm Pauli} = \proj{0}/2,\;\sigma_{1|1}^{\rm Pauli} = \proj{1}/2.\label{Eq: sigma Pauli example 2}
\end{align}
In this case, we have $a,x=0,1$ and $P(a,x)=1/2$ for every $a,x$. 
Then, demonstrating Eq.~\eqref{Eq:comp--1} is equivalent to finding some Hamiltonians $H_{a|x}$'s achieving
\begin{align}\label{Eq:12}
\max_{{\bm\tau}\in{\bf LHS}\left({\bm\sigma}^{\rm Pauli}\right)}\sum_{a,x}
\frac{\Wdiff(\hat{\tau}_{a|x},H_{a|x})}{2}<\sum_{a,x}
\frac{\Wdiff(\hat{\sigma}_{a|x}^{\rm Pauli},H_{a|x})}{2},
\end{align}
which is for $\gamma = \id/2$.
Using a steering inequality detailed in Supplemental Material IX, we choose Hamiltonians ($a=0,1$)
\begin{align}\label{Eq:13}
H_{a|0}^{\rm Pauli}=(-1)^ak_BT\delta\times X\;\;\&\;\; H_{a|1}^{\rm Pauli}=(-1)^ak_BT\delta\times Z,
\end{align}
which are Pauli $X,Z$ observables in the energy scale $k_BT\delta$ for some parameter $0<\delta<\infty$.
When $\gamma=\id/2$, we have $\WdiffSA({\bm\sigma}^{\rm Pauli},{\bf H}^{\rm Pauli}) = k_BT\delta$, and we can choose $\eta\le\delta$ in Result~\ref{Result:WorkExtraction}.
Now, for every $0<\delta<\infty$ and $0<T<\infty$, we have (see Supplemental Material IX for details)
\begin{align}\label{Eq:example classical bound}
\max_{{\bm\tau}\in{\bf LHS}\left({\bm\sigma}^{\rm Pauli}\right)}\sum_{a,x}
\frac{\Wdiff(\hat{\tau}_{a|x},H_{a|x}^{\rm Pauli})}{2}\le k_BT\left[\sqrt{2}\delta + 2\ln(\cosh\delta)\right],
\end{align}
which is strictly less than
\begin{align}\label{Eq:example quantum bound}
\sum_{a,x}
\frac{\Wdiff(\hat{\sigma}_{a|x},H_{a|x}^{\rm Pauli})}{2} = 2k_BT\left[\delta + \ln(\cosh\delta)\right].
\end{align}
Hence, to certify ${\bm\sigma}^{\rm Pauli}$'s work extraction advantage [over the set ${\bf LHS}\left({\bm\sigma}^{\rm Pauli}\right)$], it suffices to choose Hamiltonians as Pauli $X,Z$ in a suitable energy scale (i.e., $k_BT\delta$)}.
\CY{To conclude the discussion, let us consider a realistic system, the {\em Nitrogen-Vacancy (NV) centre} (see, e.g., Refs.~\cite{Awschalom2018NP,Weng2023ACSPhotonics,Neumann2010Science,Doherty2013PR,Maze2011NJP}), as an example. 
It is one of the most well-studied quantum platforms, where sequential Pauli $X,Z$ measurements on NV nuclear spin state (e.g., N14) are experimentally implementable~\cite{Neumann2010Science}.
By taking \mbox{$\delta = \eta = 1.59976\times10^{-7}$} and room temperature \mbox{$T = 300K$}, we obtain \mbox{$k_BT\delta\approx4.1357\times10^{-9}{\rm eV}$ ($\approx1{\rm MHz}$)}, which is within the NV nuclear spin's energy scale (typically 0.1MHz to 10MHz~\cite{Doherty2013PR,Maze2011NJP}).
The classical bound [Eq.~\eqref{Eq:example classical bound}] is \mbox{$\sqrt{2}\times k_BT\delta\approx5.8479\times10^{-9}{\rm eV}$}, and the quantum bound [Eq.~\eqref{Eq:example quantum bound}] is \mbox{$2\times k_BT\delta\approx8.2714\times10^{-9}{\rm eV}$}.
Their difference is $2.4235\times10^{-9}{\rm eV}$ ($\approx0.5860{\rm MHz}$), which is still within the NV nuclear spin's energy scale.}

\newpage
\section{Supplemental Material}

\CYtwo{
\subsection*{Supplemental Material I: Thermodynamic role of the function $h$}
Here, we illustrate how to obtain a specific function $h$ by imposing thermodynamic conditions.
Consider a qubit system 
% with a given thermal state 
% \CYthree{To start with, consider a fixed qubit thermal state}
\CYthree{subject to a Hamiltonian and a background temperature so that the thermal equilibrium is described by the thermal state}
\begin{align}
\gamma = p\proj{0}+(1-p)\proj{1}.
\end{align}
Then, by imposing the quantum detailed balance condition~\cite{Agarwal1973} and Gibbs-preserving condition, 
a continuous time evolution starting at identity can be described by the {\em Davies map}~\cite{Davies1974}.
More precisely, subject to the above conditions, 
Eq.~(3.13) in Ref.~\cite{Roga2010} (see also the calculation on pages 322 and 323 in Ref.~\cite{Roga2010}) gives the single-qubit Davies map as follows:
\begin{align}
\Lambda^{\rm Davies}_t(\cdot)&\coloneqq \left[1-(1-p)\left(1-e^{-At}\right)\right]\proj{0}(\cdot)\proj{0}\nonumber\\
&\quad+p\left(1-e^{-At}\right)\ket{0}\bra{1}(\cdot)\ket{1}\bra{0}\nonumber\\
&\quad+(1-p)\left(1-e^{-At}\right)\ket{1}\bra{0}(\cdot)\ket{0}\bra{1}\nonumber\\
&\quad+\left[1-p\left(1-e^{-At}\right)\right]\proj{1}(\cdot)\proj{1}\nonumber\\
&\quad+e^{-\Gamma t}\left[\proj{0}(\cdot)\proj{1}+\proj{1}(\cdot)\proj{0}\right],
\end{align}
where $A,\Gamma$ are parameters subject to $0\le A/2\le\Gamma$.
In the special case that $\Gamma = A$, a direct computation shows that
\begin{align}
\Lambda^{\rm Davies}_t(\cdot) = e^{-\Gamma t}\mathcal{I}(\cdot) + \left(1-e^{-\Gamma t}\right){\rm tr}(\cdot)\gamma,
\end{align}
which corresponds to the function $h(t) = e^{-\Gamma t}$, i.e., the partial thermalisation model with thermalistaion time scale $1/\Gamma$.
Note that the above form can also be obtained by setting $a=(1-p)(1-c)$ in the qubit state $\rho'$ defined at the beginning of page 323 in Ref.~\cite{Roga2010}.
}

\CY{
\subsection*{Supplemental Material II: Proof of Result~\ref{Eq:obs1}}
\begin{proof}
To start with, let $y_* = h\left(t_{\rm min}^{(h)}(\mathbfcal{E})\right)$.
Then, according to the definition of $t_{\rm min}^{(h)}$ \CY{[Eq.~\eqref{Eq:thermalisation time} in the main text]} and $\mathcal{D}_t^{(h)}$ \CY{[Eq.~\eqref{Eq:ThermalisationModel} in the main text]}, we have that
\begin{align}
&\left\{y_*\mathcal{E}_{a|x}(\gamma) + (1-y_*){\rm tr}\left[\mathcal{E}_{a|x}(\gamma)\right]\gamma\right\}_{a,x}\in{\bf LHS};\\
&0< y_*\le1.
\end{align}
Note that $0<y_*$ is due to the fact that $h(t):[0,\infty)\to[0,1]$ is strictly decreasing and $\lim_{t\to\infty}h(t)=0$; namely, $h(t)>0$ $\forall\,t<\infty$.
By comparing with the optimisation used to define $2^{{\rm SR}_\gamma[\mathbfcal{E}(\gamma)]}$ \CY{[Eq.~\eqref{Eq:SRgamma} in the main text]}, we conclude that $1/y_*$ is a feasible solution.
This means that $
1/y_*\ge2^{{\rm SR}_\gamma[\mathbfcal{E}(\gamma)]}$; in other words,
\begin{align}\label{Eq: tmin < SRgamma}
h\left(t_{\rm min}^{(h)}(\mathbfcal{E})\right)=y_*\le2^{-{\rm SR}_\gamma[\mathbfcal{E}(\gamma)]}.
\end{align}
It remains to prove the opposite inequality.
First, if $2^{{\rm SR}_\gamma[\mathbfcal{E}(\gamma)]}=\infty$, we have 
$2^{-{\rm SR}_\gamma[\mathbfcal{E}(\gamma)]}=0\le h\left(t_{\rm min}^{(h)}(\mathbfcal{E})\right)$.
Hence, it suffices to consider $2^{{\rm SR}_\gamma[\mathbfcal{E}(\gamma)]}<\infty$. 
To do so, let $1/p_* = 2^{{\rm SR}_\gamma[\mathbfcal{E}(\gamma)]}$, where $0< p_*\le1$ (i.e., $1/p_*$ is finite).
Then, by the definition of $2^{{\rm SR}_\gamma[\mathbfcal{E}(\gamma)]}$ \CY{[Eq.~\eqref{Eq:SRgamma} in the main text]}, we have that
\begin{align}
\left\{p_*\mathcal{E}_{a|x}(\gamma) + (1-p_*){\rm tr}\left[\mathcal{E}_{a|x}(\gamma)\right]\gamma\right\}_{a,x}\in{\bf LHS}.
\end{align}
Now, recall that $h(t):[0,\infty)\to[0,1]$ is strictly decreasing, continuous, with $h(0)=1$ and $\lim_{t\to\infty}h(t)=0$.
This means that the function $h(t)$ is one-to-one and onto (i.e., a bijection) between the intervals $[0,\infty)$ and $(0,1]$, and its inverse function exists.
In particular, for $p_*\in(0,1]$, there exists a unique $t_*\in[0,\infty)$ such that
\begin{align}
h(t_*) = p_*.
\end{align}
Consequently, 
\begin{align}
\left\{h(t_*)\mathcal{E}_{a|x}(\gamma) + [1-h(t_*)]{\rm tr}\left[\mathcal{E}_{a|x}(\gamma)\right]\gamma\right\}_{a,x}\in{\bf LHS}.
\end{align}
In other words, $t_*$ is a feasible solution to the minimisation defining $t_{\rm min}^{(h)}(\mathbfcal{E})$ \CY{[Eq.~\eqref{Eq:thermalisation time} in the main text]}.
This implies that $t_*\ge t_{\rm min}^{(h)}(\mathbfcal{E})$.
Since $h(t)$ is strictly decreasing, we obtain
\begin{align}\label{Eq: tmin > SRgamma}
2^{-{\rm SR}_\gamma[\mathbfcal{E}(\gamma)]} = p_* = h(t_*)\le h\left(t_{\rm min}^{(h)}(\mathbfcal{E})\right).
\end{align}
The proof is completed by combining Eqs.~\eqref{Eq: tmin < SRgamma} and~\eqref{Eq: tmin > SRgamma}.
\end{proof}
}

\CYtwo{
\subsection*{Supplemental Material III: Proof of Result~\ref{result: environment result}}
}
\CYtwo{
Here, we state the complete version of Result~\ref{result: environment result} in the main text as the following theorem.
We use $d_S,d_E,d_{E'}$ to denote the dimensions of the systems $S,E,E'$.
Also, subscripts denote the corresponding (sub-)system, and $\ket{\Phi^+}_{SE}\coloneqq\sum_{i=0}^{d-1}(1/\sqrt{d})\ket{ii}_{SE}$ is a maximally entangled state in $SE$ with the given computational bases $\{\ket{i}_S\}_{i=0}^{d_S-1}, \{\ket{i}_E\}_{i=0}^{d_E-1}$.\\

\begin{result}\label{result: environment lemma}
Let $\mathbfcal{E}$ be a collection of instruments in a system $S$.
The following statements are equivalent:
\begin{enumerate}
\item\label{condition 1}
$\mathbfcal{E}$ is compatible.
\item\label{condition 2 plus}
$\{(\mathcal{E}_{a|x}\otimes\mathcal{I}_E)(\CYthree{{\gamma}_{SE}})\}_{a,x}\in{\bf LHS}$
for every auxiliary system $E$ with $d_E<\infty$ and \CYthree{${\gamma}_{SE}$} (which can be non-full-rank).
\item\label{condition 2}
$\{(\mathcal{E}_{a|x}\otimes\mathcal{I}_E)(\gamma_{SE})\}_{a,x}\in{\bf LHS}$
for every auxiliary system $E$ with $d_E\le d_S$ and full-rank $\gamma_{SE}$.
\item\label{condition 2 minus}
$\{(\mathcal{E}_{a|x}\otimes\mathcal{I}_E)(\gamma_{SE}^{(\epsilon)})\}_{a,x}\in{\bf LHS}$ $\forall\,0<\epsilon\le1$,
where $E$ is an auxiliary system with $d_E=d_S=d$, and 
\begin{align}\label{Eq:iso}
\gamma_{SE}^{(\epsilon)} \coloneqq (1-\epsilon)\proj{\Phi^+}_{SE} + \epsilon\id_{SE}/d^2.
\end{align}
\end{enumerate}
\end{result}
We remark that the one-parameter family of states defined in Eq.~\eqref{Eq:iso} is known as {\em isotropic states}~\cite{HorodeckiPRA1999}, which is equivalent to {\em Werner states}~\cite{WernerPRA} in two-qubit systems. 
They are of great importance in the study of nonlocality and steering due to their simple structure (see, e.g., Refs.~\cite{WernerPRA,HorodeckiPRA1999,HorodeckiPRA1999-2,Hsieh2016PRA,Quintino2016,CavalcantiPRA2013,Hsieh2020PRR,Hsieh2021PRXQ,ZhangPRL2024,RennerPRL2024,HsiehPRA(E)2018} for their applications).
Importantly, the above result implies that {\em full-rank} isotropic states form a universal family of states that can certify {\em all possible} incompatible instruments.
\begin{proof}
Using definitions of compatible instruments [Eq.~\eqref{Eq:instrument-JM} in the main text] and ${\bf LHS}$ [Eq.~\eqref{Eq:LHS-definition} in the main text], we conclude that statement~\ref{condition 1} $\Rightarrow$ statement~\ref{condition 2 plus} $\Rightarrow$ statement~\ref{condition 2} $\Rightarrow$ statement~\ref{condition 2 minus}.
Hence, it suffices to show that statement~\ref{condition 2 minus} implies statement~\ref{condition 1} to complete the proof.
Suppose statement~\ref{condition 2 minus} holds.
Using Eq.~\eqref{Eq:iso}, we have, for every $a,x$ and $0<\epsilon\le1$,
\begin{align}
(\mathcal{E}_{a|x}\otimes\mathcal{I}_E)\left(\gamma_{SE}^{(\epsilon)}\right) &= (1-\epsilon)(\mathcal{E}_{a|x}\otimes\mathcal{I}_E)\left(\proj{\Phi^+}_{SE}\right)\nonumber\\
&\quad+ \epsilon(\mathcal{E}_{a|x}\otimes\mathcal{I}_E)\left(\id_{SE}/d^2\right)
\end{align}
Hence, statement~\ref{condition 2 minus} implies that the state assemblage \mbox{$\{(\mathcal{E}_{a|x}\otimes\mathcal{I}_E)(\proj{\Phi^+}_{SE})\}_{a,x}$} can be arbitrarily close to the set ${\bf LHS}$ by setting $\epsilon\to0$.
In other words, it is in the closure of ${\bf LHS}$.
Since ${\bf LHS}$ is a compact set (and hence a closed set), it is identical to its closure. 
We conclude that
\begin{align}
\{(\mathcal{E}_{a|x}\otimes\mathcal{I}_E)(\proj{\Phi^+}_{SE})\}_{a,x}\in{\bf LHS}.
\end{align}
By the definition of ${\bf LHS}$ [Eq.~\eqref{Eq:LHS-definition} in the main text], there exist some (conditional) probability distributions $\{P_\lambda\}_\lambda,\{P(a|x,\lambda)\}_{a,x,\lambda}$ and bipartite states $\{\rho_{\lambda,SE}\}_\lambda$ in $SE$ such that
\begin{align}\label{Eq:computation choi 0001}
(\mathcal{E}_{a|x}\otimes\mathcal{I}_E)(\proj{\Phi^+}_{SE}) = \sum_\lambda P(a|x,\lambda)P_\lambda\rho_{\lambda,SE}.
\end{align}
For each $\lambda$, define the following linear map in $S$ (note that, as stated in statement~\ref{condition 2 minus}, we have $d_S=d_E=d$):
\begin{align}\label{Eq:G lambda def}
\mathcal{G}_\lambda(\cdot)\coloneqq {\rm tr}_{E'}\left([\id_S\otimes(\cdot)_{E'}]\rho_{\lambda,SE'}^{t_{E'}}\right)P_\lambda d,
\end{align}
where $(\cdot)^{t_{E'}}\coloneqq\sum_{n,m}\ket{n}\bra{m}_{E'}(\cdot)_{E'}\ket{n}\bra{m}_{E'}$ is the transpose operation in the system $E'$ (with a given computational basis $\{\ket{n}_{E'}\}_{n=0}^{d-1}$), and $E'$ is another auxiliary system with dimension $d$.
Then, a direct computation shows that
\begin{align}\label{Eq:G lambda Choi definition}
(\mathcal{G}_\lambda\otimes\mathcal{I}_E)(\proj{\Phi^+}_{SE})= P_\lambda\rho_{\lambda,SE}\quad\forall\,\lambda,
\end{align}
where we have used the fact $(\bra{\Phi^+}_{EE'}\otimes\id_S)(P_{SE}\otimes\id_{E'}) = (\bra{\Phi^+}_{EE'}\otimes\id_S)(\id_{E}\otimes P_{SE'}^{t_{E'}})$ which holds for every bipartite operator $P$.
Eq.~\eqref{Eq:G lambda Choi definition} is the so-called {\em Choi-Jamio\l kowski isomorphism}~\cite{Choi1975,Jamiolkowski1972}, which maps the map $\mathcal{G}_\lambda$ to a bipartite operator.
In this language, Eq.~\eqref{Eq:G lambda def} is the so-called {\em inverse Choi map}, as it maps a bipartite operator to a linear map.
Combining Eqs.~\eqref{Eq:computation choi 0001} and~\eqref{Eq:G lambda Choi definition}, we obtain
\begin{align}\label{Eq:computation choi 0002}
&\left(\mathcal{E}_{a|x}\otimes\mathcal{I}_E\right)(\proj{\Phi^+}_{SE})=\nonumber\\
&\quad\quad\left[\left(\sum_\lambda P(a|x,\lambda)\mathcal{G}_\lambda\right)\otimes\mathcal{I}_E\right](\proj{\Phi^+}_{SE})\quad\forall\,a,x.
\end{align}
Now, by the Choi-Jamio\l kowski isomorphism, two linear maps $\mathcal{N}_S,\mathcal{M}_S$ in $S$ are identical to each other if and only if $(\mathcal{N}_S\otimes\mathcal{I}_E)(\proj{\Phi^+}_{SE})=(\mathcal{M}_S\otimes\mathcal{I}_E)(\proj{\Phi^+}_{SE})$~\cite{Choi1975,Jamiolkowski1972}.
Hence, Eq.~\eqref{Eq:computation choi 0002} is equivalent to
\begin{align}\label{Eq: computation choi 0003}
\mathcal{E}_{a|x}=\sum_\lambda P(a|x,\lambda)\mathcal{G}_\lambda\quad\forall\,a,x.
\end{align}
It remains to show that $\{\mathcal{G}_\lambda\}_\lambda$ is a valid instrument.
By summing over $a$ in Eq.~\eqref{Eq: computation choi 0003}, we obtain
% \begin{align}
$
\sum_\lambda\mathcal{G}_\lambda = \sum_a\mathcal{E}_{a|x},
$
% \end{align}
which is trace-preserving.
Hence, it remains to check whether $\mathcal{G}_\lambda$'s are trace-non-increasing and completely-positive.
By the Choi-Jamio\l kowski isomorphism, a linear map $\mathcal{N}_S$ in $S$ is trace-preserving if and only if ${\rm tr}_S[(\mathcal{N}_S\otimes\mathcal{I}_E)(\proj{\Phi^+}_{SE})] = \id_E/d$~\cite{Choi1975,Jamiolkowski1972} (see also, e.g., Theorem 2.1.1 in Ref.~\cite{Hsieh-thesis}).
With Eq.~\eqref{Eq:computation choi 0001}, we obtain
\begin{align}
{\rm tr}_S\left(\sum_\lambda P_\lambda\rho_{\lambda,SE}\right)& = {\rm tr}_S\left[\left(\sum_a\mathcal{E}_{a|x}\otimes\mathcal{I}_E\right)(\proj{\Phi^+}_{SE})\right]\nonumber\\
&= \id_E/d.
\end{align}
Together with Eq.~\eqref{Eq:G lambda Choi definition}, we conclude that, for each $\lambda$, 
\begin{align}
{\rm tr}_S\left[(\mathcal{G}_\lambda\otimes\mathcal{I}_E)(\proj{\Phi^+}_{SE})\right] = {\rm tr}_S(P_\lambda\rho_{\lambda,SE})\le\id_E/d.\label{Eq:TNI}
\end{align}
For an arbitrary input state $\omega$, a direct computation shows that
\begin{align}
{\rm tr}[\mathcal{G}_\lambda(\omega)] &= d{\rm tr}\left[(\id_S\otimes\omega_{E'}^{t_{E'}})(\mathcal{G}_\lambda\otimes\mathcal{I}_{E'})(\proj{\Phi^+}_{SE'})\right]\nonumber\\
&\le d{\rm tr}\left(\omega_{E'}^{t_{E'}}\times\id_{E'}/d\right) = {\rm tr}(\omega),
\end{align}
where we have used Eq.~\eqref{Eq:TNI} to obtain the inequality.
Hence, $\mathcal{G}_\lambda$'s are trace-non-increasing.
Finally, to show that $\mathcal{G}_\lambda$'s are completely-positive, for each $\lambda$, Eq.~\eqref{Eq:G lambda Choi definition} implies 
\begin{align}
(\mathcal{G}_\lambda\otimes\mathcal{I}_E)(\proj{\Phi^+}_{SE})\ge0.\label{Eq:CP}
\end{align}
Again, by the Choi-Jamio\l kowski isomorphism, a linear map $\mathcal{N}_S$ in $S$ is completely-positive if and only if \mbox{$(\mathcal{N}_S\otimes\mathcal{I}_E)(\proj{\Phi^+}_{SE})\ge0$}~\cite{Choi1975,Jamiolkowski1972} (see, e.g., Theorem 2.1.1 and Lemma A.1.1 in Ref.~\cite{Hsieh-thesis}).
Equation~\eqref{Eq:CP} means that $\mathcal{G}_\lambda$'s are completely-positive, and, consequently, $\{\mathcal{G}_\lambda\}_\lambda$ is indeed a valid instrument. 
Then, $\mathbfcal{E}$ is compatible since Eq.~\eqref{Eq: computation choi 0003} is a valid decomposition for the definition [Eq.~\eqref{Eq:instrument-JM} in the main text].
The proof is completed.
\end{proof}
}

\subsection*{\CYnew{Supplemental Material IV:} ${\rm SR}_\gamma$ as an SDP}\label{APP:SRgammaSDP}
First of all, for an arbitrarily given ${\bm\tau}\in{\bf LHS}$, we can write
\begin{align}
\tau_{a|x} = \sum_{i=1}^{|{\bf i}|}q_i D(a|x,i)\rho_i\quad\forall\,a,x,
\end{align}
which sums over all possible {\em deterministic probability distributions} $D(a|x,i)$, where each of them assigns exactly one output value $a$ for each input $x$~\cite{UolaRMP2020,Cavalcanti2016}.
In other words, each $D(a|x,i)$ is a deterministic mapping from \mbox{$x=0,1,...,|{\bf x}|-1$} to \mbox{$a=0,1,...,|{\bf a}|-1$}, and there are $|{\bf i}|\coloneqq|{\bf a}|^{|{\bf x}|}$ many of them.
By defining $\eta_i = q_i\rho_i$, we obtain a simple useful fact~\cite{UolaRMP2020,Cavalcanti2016}:
\begin{fact}\label{Eq:LHSfact}
Every ${\bm\tau}\in{\bf LHS}$ can be expressed by 
\begin{align}
\tau_{a|x} = \sum_{i=1}^{|{\bf i}|}D(a|x,i)\eta_i\quad\forall\,a,x
\end{align}
with $\eta_i\ge0$ and $\sum_{i=1}^{|{\bf i}|}{\rm tr}(\eta_i)=1$.
\end{fact}
Using the above fact and \CYnew{Eq.~\eqref{Eq:SRgamma} in the main text}, we can write $2^{-{\rm SR}_\gamma({\bm\sigma})}$ into the following form [since now, we define $p_{a|x} \coloneqq {\rm tr}(\sigma_{a|x})$; also, $\sum_i$ is the abbreviation of $\sum_{i=1}^{|{\bf i}|}$]:
\begin{equation}\label{Eq:SRgammaSDPprimal}
\begin{split}
\max_{\{\eta_i\}_i,q}\quad&q\\
{\rm s.t.}\quad&0\le q\le1;\;\;\eta_i\ge0\;\forall\,i;\\
&\sum_i D(a|x,i)\eta_i+q(\gamma p_{a|x}-\sigma_{a|x}) = \gamma p_{a|x}\;\forall\,a,x.
\end{split}
\end{equation}
Note that the constraint $\sum_i{\rm tr}(\eta_i)=1$ is already guaranteed by taking trace and summing over $a$ of the constraint \mbox{$\sum_i D(a|x,i)\eta_i+q(\gamma p_{a|x}-\sigma_{a|x}) = \gamma p_{a|x}$.}
Now, define
\begin{align}
{\bf V}&\coloneqq\left(\bigoplus_i\eta_i\right)\oplus q,\nonumber\\
{\bf B}&\coloneqq\bigoplus_{a,x}\left(\gamma p_{a|x}\right),\nonumber\\
\Phi({\bf V})&\coloneqq\bigoplus_{a,x}\left[\sum_iD(a|x,i)\eta_i+q(\gamma p_{a|x}-\sigma_{a|x})\right],
\end{align}
where $\Phi$ is a Hermitian-preserving linear map~\cite{watrous_2018}. 
Then we have
\begin{equation}\label{Eq:SDPprimal}
\begin{split}
2^{-{\rm SR}_\gamma({\bm\sigma})} = \max_{{\bf V}}\quad&q\\
{\rm s.t.}\quad& \Phi({\bf V}) = {\bf B};\;\;q \le 1;\;\;{\bf V}\ge0.
\end{split}
\end{equation}
This is in the standard form of the primal problem of SDP, whose {\em dual problem} can be found as (see, e.g., Section 1.2.3 in Ref.~\cite{watrous_2018}; see also Ref.~\cite{SDP-textbook})
\begin{equation}\label{Eq:SDPdual}
\begin{split}
\min_{{\bf\Ydual},\omega}\quad&\sum_{a,x}{\rm tr}\left(\Ydual_{a|x}\gamma\right)p_{a|x}+\omega\\
{\rm s.t.}\quad&\Ydual_{a|x} = \Ydual_{a|x}^\dagger\;\forall\,a,x;\;\;\omega\ge0;\\
& \Phi^\dagger({\bf\Ydual})+\left(\bigoplus_i0\right)\oplus \omega \ge \left(\bigoplus_i0\right)\oplus 1,
\end{split}
\end{equation}
where ${\bf\Ydual} = \bigoplus_{a,x}\Ydual_{a|x}$.
It remains to find $\Phi^\dagger({\bf\Ydual})$.
To this end, note that we have
\begin{align}
\sum_{a,x}{\rm tr}\left(\Ydual_{a|x}\Phi({\bf V})_{a|x}\right)=&\sum_i{\rm tr}\left[\sum_{a,x}\Ydual_{a|x}D(a|x,i)\eta_i\right]\nonumber\\
&+ {\rm tr}\left[\sum_{a,x}\Ydual_{a|x}(\gamma p_{a|x}-\sigma_{a|x})\right]\times q.
\end{align}
% $
Hence, using the definition of inner product, we conclude that
$
\Phi^\dagger({\bf\Ydual}) = \left[\bigoplus_i\left(\sum_{a,x}\Ydual_{a|x}D(a|x,i)\right)\right]\oplus{\rm tr}\left[\sum_{a,x}\Ydual_{a|x}(\gamma p_{a|x}-\sigma_{a|x})\right],
$
and we obtain the dual form of $2^{-{\rm SR}_\gamma({\bm\sigma})}$ as follows:
\begin{equation}\label{Eq:SRgammaSDPdual}
\begin{split} 
\min_{{\bf\Ydual},\omega}\quad&\sum_{a,x}{\rm tr}(\Ydual_{a|x}\gamma)p_{a|x} + \omega\\
{\rm s.t.}\quad& \Ydual_{a|x} = \Ydual_{a|x}^\dagger\;\forall\,a,x;\;\;\omega\ge0;\\
&\sum_{a,x}D(a|x,i)\Ydual_{a|x}\ge0\;\forall\,i;\\
&\sum_{a,x}{\rm tr}(\Ydual_{a|x}\gamma)p_{a|x}+\omega\ge\sum_{a,x}{\rm tr}(\Ydual_{a|x}\sigma_{a|x})+1.
\end{split}
\end{equation}
Note that this is {\em equal to} $2^{-{\rm SR}_\gamma({\bm\sigma})}$ since the {\em strong duality} holds~\cite{watrous_2018} --- this can be checked by noting that the primal problem is finite and feasible (e.g., by simply considering $q=0$), and the dual problem is strictly feasible; namely, by choosing $\omega>1$ and $\Ydual_{a|x} = \id$ for every $a,x$, one can achieve $\sum_{a,x}D(a|x,i)\Ydual_{a|x}>0\;\forall\,i$ and $\sum_{a,x}{\rm tr}(\Ydual_{a|x}\gamma)p_{a|x}+\omega>\sum_{a,x}{\rm tr}(\Ydual_{a|x}\sigma_{a|x})+1$ (see, e.g., Theorem~1.18 in Ref.~\cite{watrous_2018}).

\subsection*{\CYnew{Supplemental Material V:} Proof of \CYnew{Result~\ref{Result:ResourceMonotone}}}\label{App:Proof-Result:ResourceMonotone}
\begin{proof}
First, if $\mathbfcal{E}$ is compatible, one can directly conclude that $t_{\rm min}^{(h)}(\mathbfcal{E})=0$ by \CYnew{Eq.~\eqref{Eq:thermalisation time} in the main text}.
Hence, it suffices to show the non-increasing property under deterministic allowed operations. 
Note that the decomposition given in the definition of ${\rm SR}_\gamma[\mathbb{F}(\mathbfcal{E})(\gamma)]$ [i.e., \CYnew{Eq.~\eqref{Eq:SRgamma} in the main text}] can be rewritten as (in what follows, we define $q = 1/\SRvar$)
\begin{align}
&q\mathbb{F}(\mathbfcal{E})_{a|x}(\gamma) + (1-q){\rm tr}[\mathbb{F}(\mathbfcal{E})_{a|x}(\gamma)]\gamma\nonumber\\
&=q\sum_{b,y}P'(a|xyb)P(y|x)\CYnew{\mathcal{Q}}\circ\mathcal{E}_{b|y}\circ\mathcal{P}(\gamma)\nonumber\\
&\quad+(1-q)\sum_{b,y}P'(a|xyb)P(y|x){\rm tr}\left[\CYnew{\mathcal{Q}}\circ\mathcal{E}_{b|y}\circ\mathcal{P}(\gamma)\right]\gamma\nonumber\\
& = \sum_{b,y}P'(a|xyb)P(y|x)\CYnew{\mathcal{Q}}\biggl(q\mathcal{E}_{b|y}(\gamma) + (1-q){\rm tr}\left[\mathcal{E}_{b|y}(\gamma)\right]\gamma\biggr),
\end{align}
where we have used the trace-preserving as well as Gibbs-preserving properties of \CYnew{$\mathcal{P},\mathcal{Q}$}.
Consequently,
$2^{-{\rm SR}_\gamma[\mathbb{F}(\mathbfcal{E})(\gamma)]}$ can be expressed as
\begin{equation}\label{Eq:non-increasing-computation001}
\begin{split}
\max_{q,{\bm\omega}}\quad&q\\
{\rm s.t.}\quad&0\le q \le 1;\;\;{\bm\omega}\in{\bf LHS};\;\;\forall a,x:\\ &\omega_{a|x} = \sum_{b,y}P'(a|xyb)P(y|x)\\
&\quad\quad\;\times\CYnew{\mathcal{Q}}\biggl(q\mathcal{E}_{b|y}(\gamma) + (1-q){\rm tr}\left[\mathcal{E}_{b|y}(\gamma)\right]\gamma\biggr).
\end{split}
\end{equation}
Note that $\left\{\sum_{b,y}P'(a|xyb)P(y|x)\CYnew{\mathcal{Q}}(\tau_{b|y})\right\}_{a,x}\in{\bf LHS}$ if \mbox{${\bm\tau}\in{\bf LHS}$}, since it is an allowed operation of steering~\cite{Gallego2015PRX}. 
Hence, \CYnew{Eq.~\eqref{Eq:non-increasing-computation001}}
is lower bounded by
\begin{equation}\label{Eq:non-increasing-computation002}
\begin{split}
\max_{q,{\bm\tau}}\quad&q\\
{\rm s.t.}\quad&0\le q \le 1;\;\;{\bm\tau}\in{\bf LHS};\\ 
&\tau_{a|x} = q\mathcal{E}_{a|x}(\gamma) + (1-q){\rm tr}\left[\mathcal{E}_{a|x}(\gamma)\right]\gamma\,\forall a,x.
\end{split}
\end{equation}
This is because a feasible solution $(q,{\bm\tau})$ of Eq.~\eqref{Eq:non-increasing-computation002} can always induce a feasible solution $(q,{\bm\omega})$ of Eq.~\eqref{Eq:non-increasing-computation001} with the same objective function (i.e., $q$).
Since Eq.~\eqref{Eq:non-increasing-computation002} is exactly $2^{-{\rm SR}_\gamma[\mathbfcal{E}(\gamma)]}$, we conclude that $2^{-{\rm SR}_\gamma[\mathbb{F}(\mathbfcal{E})(\gamma)]}\ge2^{-{\rm SR}_\gamma[\mathbfcal{E}(\gamma)]}$,
namely, we obtain
\begin{align}
{\rm SR}_\gamma[\mathbb{F}(\mathbfcal{E})(\gamma)]\le{\rm SR}_\gamma[\mathbfcal{E}(\gamma)].
\end{align}
Finally, by using Eqs.~\eqref{Eq:thermalisation time},~\eqref{Eq:SRgamma} in the main text and the fact that \CYnew{the inverse function} $h^{-1}(y):[0,1]\to\mathbb{R}$ is again strictly decreasing, we conclude that
\begin{align}\label{Eq:tmin non-increasing allowed op}
t_{\rm min}^{(h)}[\mathbb{F}(\mathbfcal{E})] &= h^{-1}\left(2^{-{\rm SR}_\gamma[\mathbb{F}(\mathbfcal{E})(\gamma)]}\right)\nonumber\\
&\le h^{-1}\left(2^{-{\rm SR}_\gamma[\mathbfcal{E}(\gamma)]}\right) = t_{\rm min}^{(h)}(\mathbfcal{E}).
\end{align}
The proof is thus completed.
\end{proof}

\CYthree{As a remark, deterministic allowed operations are closed under composition.
To see this, consider two deterministic allowed operations (specified by $i=1,2$)
\begin{align}
\mathbb{F}_{(i)}(\mathbfcal{E})_{a|x}\coloneqq\sum_{b,y}P_{(i)}'(a|xyb)P_{(i)}(y|x)\mathcal{Q}_{(i)}\circ\mathcal{E}_{b|y}\circ\mathcal{P}_{(i)}.
\end{align} 
Then we have, for every $\mathbfcal{E}$,
\begin{align}
&\left[\left(\mathbb{F}_{(2)}\circ\mathbb{F}_{(1)}\right)(\mathbfcal{E})\right]_{a|x}\nonumber\\
% &= \sum_{b,y}P_{(2)}'(a|xyb)P_{(2)}(y|x)\mathcal{Q}_{(2)}\circ\mathbb{F}_{(1)}(\mathbfcal{E})_{b|y}\circ\mathcal{P}_{(2)}\nonumber\\
&=\sum_{c,z}\sum_{b,y}P_{(2)}'(a|xyb)P_{(1)}'(b|yzc)P_{(1)}(z|y)P_{(2)}(y|x)\nonumber\\
&\quad\quad\quad\quad\times\mathcal{Q}_{(2)}\circ\mathcal{Q}_{(1)}\circ\mathcal{E}_{c|z}\circ\mathcal{P}_{(1)}\circ\mathcal{P}_{(2)}\nonumber\\
&=\sum_{c,z}P'_{\rm sum}(a|xzc)P_{\rm sum}(z|x)\mathcal{Q}_{\rm sum}\circ\mathcal{E}_{c|z}\circ\mathcal{P}_{\rm sum},
\end{align}
where $\mathcal{Q}_{\rm sum}\coloneqq\mathcal{Q}_{(2)}\circ\mathcal{Q}_{(1)}$, $\mathcal{P}_{\rm sum}\coloneqq\mathcal{P}_{(1)}\circ\mathcal{P}_{(2)}$, and
\begin{align}
P_{\rm sum}(z|x)&\coloneqq\sum_{y}P_{(1)}(z|y)P_{(2)}(y|x),\\
P'_{\rm sum}(a|xzc)&\coloneqq\sum_{b,y}\frac{P_{(2)}'(a|xyb)P_{(1)}'(b|yzc)P_{(1)}(z|y)P_{(2)}(y|x)}{P_{\rm sum}(z|x)}
% \times\sum_bP_{(2)}'(a|xyb)P_{(1)}'(b|yzc).
\end{align}
are both valid conditional probability distributions.
Since Gibbs-preserving channels are closed under composition, 
% the result follows; namely, 
$\mathbb{F}_{(2)}\circ\mathbb{F}_{(1)}$ is again a deterministic allowed operation.
}

\subsection*{\CYnew{Supplemental Material VI:} Proof of \CYnew{Result~\ref{Result:No-extension}}}\label{App:Proof-Result:No-extension}
\begin{proof}
First, let us write ${\bm\sigma} = \mathbfcal{E}(\gamma)$.
Since $\mathbfcal{E}$'s average channel is Gibbs-preserving, we have $\sum_a\mathcal{E}_{a|x}(\gamma) = \gamma$.
Namely, ${\bm\sigma}$'s reduced state is $\gamma$.
Suppose that there exists an allowed stochastic operation (i.e., an ${\rm LF_1}$ filter with the Kraus operator $K$) achieving 
% $
\begin{align}
\sigma_{a|x}\mapsto\omega_{a|x} \coloneqq K\sigma_{a|x} K^\dagger / p_\gamma,
\end{align}
% $
where, again, $p_\gamma = {\rm tr}\left(K\gamma K^\dagger\right)$.
Since it is a stochastic allowed operation, we have the following two conditions [corresponding to conditions (i) and (ii) in the main text]:
\begin{align}\label{Eq:condition(i)}
&{\rm tr}(\sigma_{a|x}) = {\rm tr}(K\sigma_{a|x}K^\dagger)/p_\gamma = {\rm tr}(\omega_{a|x})\quad\forall\,a,x;\\
&K\gamma K^\dagger/p_\gamma = \gamma.\label{Eq:condition(ii)}
\end{align}
From here, we observe that ${\bm\omega}$'s reduced state is $\gamma$.
Now, we recall a recent result from stochastic steering distillation: 
\begin{lemma}{\em\cite{Ku2022NC,Hsieh2023}}
Consider two state assemblages ${\bm\kappa},{\bm\tau}$.
If ${\bm\kappa}$ can be converted to ${\bm\tau}$ by some ${\rm LF_1}$ filter, then there exists a unitary $U$ with ${\rm supp}(\rho_{\bm\tau})\subseteq{\rm supp}(U\rho_{\bm\kappa}U^\dagger)$ such that
\begin{align}
\tau_{a|x} = \sqrt{\rho_{\bm\tau}}U\sqrt{\rho_{\bm\kappa}}^{\;-1}\kappa_{a|x}\sqrt{\rho_{\bm\kappa}}^{\;-1}U^\dagger\sqrt{\rho_{\bm\tau}}\quad\forall\,a,x,
\end{align}
where $\rho_{\bm\kappa}$ ($\rho_{\bm\tau}$) is the reduced state of ${\bm\kappa}$ (${\bm\tau}$).
\end{lemma}
Setting ${\bm\kappa}={\bm\sigma}$ and ${\bm\tau}={\bm\omega}$ in the above lemma, we obtain
\begin{align}\label{Eq:omega sigma relation}
\omega_{a|x} = \sqrt{\gamma}U\sqrt{\gamma}^{\;-1}\sigma_{a|x}\sqrt{\gamma}^{\;-1}U^\dagger\sqrt{\gamma}\quad\forall\,a,x,
\end{align}
for some unitary $U$ (note that we always assume $\gamma$ is full-rank).
This means that, after the given ${\rm LF_1}$ filter,
\begin{align}\label{Eq:no-go proof final computation}
&2^{-{\rm SR}_\gamma({\bm\omega})} = \max\left\{0\le q\le1\,\middle|\,q\omega_{a|x} + (1-q){\rm tr}(\omega_{a|x})\gamma\in{\bf LHS}\right\}\nonumber\\
&\ge\max\left\{0\le q\le1\,\middle|\,q\sigma_{a|x} + (1-q){\rm tr}(\omega_{a|x})\gamma\in{\bf LHS}\right\}\nonumber\\
&=\max\left\{0\le q\le1\,\middle|\,q\sigma_{a|x} + (1-q){\rm tr}(\sigma_{a|x})\gamma\in{\bf LHS}\right\}\nonumber\\
&= 2^{-{\rm SR}_\gamma({\bm\sigma})}.
\end{align}
To see the inequality, note that if the state assemblage \mbox{$q\sigma_{a|x} + (1-q){\rm tr}(\omega_{a|x})\gamma\in{\bf LHS}$}, 
% then $\sum_{a}\theta_{a|x} = \gamma$; i.e., it is a LHS state assemblage with reduced state $\gamma$.
then one can check that 
% $\theta'_{a|x} =
\begin{align}
&q\omega_{a|x} + (1-q){\rm tr}(\omega_{a|x})\gamma=\nonumber\\
&\;\left(\sqrt{\gamma}U\sqrt{\gamma}^{\;-1}\right)\left[q\sigma_{a|x} + (1-q){\rm tr}(\omega_{a|x})\gamma\right]\left(\sqrt{\gamma}^{\;-1}U^\dagger\sqrt{\gamma}\right)
% $
\end{align}
is again in ${\bf LHS}$~\cite{Ku2022NC,Hsieh2023}.
% By Eq.~\eqref{Eq:omega sigma relation}, we have $\theta'_{a|x} = q\omega_{a|x} + (1-q){\rm tr}(\omega_{a|x})\gamma$, which is a valid decomposition for the maximisation in the first line of Eq.~\eqref{Eq:no-go proof final computation}.
Hence, in Eq.~\eqref{Eq:no-go proof final computation}, the maximisation in the second line is sub-optimal to the first one.
We thus obtain ${\rm SR}_\gamma({\bm\omega})\le{\rm SR}_\gamma({\bm\sigma})$; that is, stochastic allowed operations cannot increase ${\rm SR}_\gamma$.
Finally, by using the same argument as in Eq.~\eqref{Eq:tmin non-increasing allowed op}, we conclude that $t_{\rm min}^{(h)}(\mathbfcal{E})$ cannot be increased.
\end{proof}

\CYthree{As a remark, stochastic allowed operations are closed under composition.
To see this, consider a given ${\bm\sigma}=\{\sigma_{a|x}\}_{a,x}$ with $\sum_a\sigma_{a|x}=\gamma$.
Then sequentially applying two stochastic allowed operations described by $K_{(1)},K_{(2)}$ gives
\begin{align}
\sigma_{a|x}&\stackrel{K_{(1)}}{\mapsto}\sigma_{a|x}^{(1)}\coloneqq K_{(1)}\sigma_{a|x} K_{(1)}^\dagger/p_{(1)}\nonumber\\
&\stackrel{K_{(2)}}{\mapsto}\sigma_{a|x}^{(2)}\coloneqq K_{(2)}K_{(1)}\sigma_{a|x} \left(K_{(2)}K_{(1)}\right)^\dagger/p_{(1)}p_{(2)},
\end{align}
where filters' success probabilities are $p_{(1)} = {\rm tr}\left(K_{(1)}\gamma K_{(1)}^\dagger\right)$ and $p_{(2)} = {\rm tr}\left(K_{(2)}\sum_a\sigma_{a|x}^{(1)} K_{(2)}^\dagger\right) = {\rm tr}\left(K_{(2)}\gamma K_{(2)}^\dagger\right)$ [Eq.~\eqref{Eq:condition(ii)} for $K_{(1)}$ is used].
Note that stochastic allowed operation's definition [more precisely, Eq.~\eqref{Eq:condition(i)}, i.e., condition (i) in the main text] can be input-dependent. 
Here, $K_{(1)}$ is for the input ${\bm\sigma}$, and $K_{(2)}$ is for the input $\{\sigma_{a|x}^{(1)}\}_{a,x}$.
Hence, Eq.~\eqref{Eq:condition(i)} reads
\begin{align}\label{Eq:for condition (i)}
{\rm tr}(\sigma_{a|x}) = {\rm tr}\left(\sigma_{a|x}^{(1)}\right)
\;\;\&\;\;{\rm tr}\left(\sigma_{a|x}^{(1)}\right) = {\rm tr}\left(\sigma_{a|x}^{(2)}\right)\quad\forall\,a,x.
\end{align}
Now, define a new filter $K_{\rm sum}\coloneqq K_{(2)}K_{(1)}$.
Using Eq.~\eqref{Eq:condition(ii)} for $K_{(1)}$ and $K_{(2)}$, we obtain
\begin{align}\label{Eq:for condition (ii)}
\sum_a\sigma_{a|x}^{(2)}=K_{(2)}K_{(1)}\gamma\left(K_{(2)}K_{(1)}\right)^\dagger/p_{(1)}p_{(2)} = \gamma.
\end{align}
Hence, $K_{\rm sum}$'s success probability (denoted by $p_{\rm sum}$) reads
% \begin{align}
$
p_{\rm sum} = {\rm tr}\left(K_{\rm sum}\gamma K_{\rm sum}^\dagger\right) = p_{(1)}p_{(2)},
$
% \end{align}
and the composition of $K_{(1)},K_{(2)}$ is equivalent to the single filter $K_{\rm sum}$; namely,
\begin{align}
\sigma_{a|x}^{(2)} = K_{\rm sum}\sigma_{a|x} K_{\rm sum}^\dagger/p_{\rm sum}\quad\forall\,a,x.
\end{align}
Finally, Eq.~\eqref{Eq:for condition (i)} means that $K_{\rm sum}$ satisfies Eq.~\eqref{Eq:condition(i)} for the input ${\bm\sigma}$, and Eq.~\eqref{Eq:for condition (ii)} implies that Eq.~\eqref{Eq:condition(ii)} is also satisfied.
This means conditions (i) and (ii) in the main text are satisfied, and $K_{\rm sum}$ is a stochastic allowed operation for the input ${\bm\sigma}$.
}

\subsection*{\CYnew{Supplemental Material VII: Relative interior and Result~\ref{Result:Strong-no-go}}}\label{Proof-Result:Strong-no-go}
\subsubsection*{Relative interior of ${\bf LHS}$}
The {\em relative interior} of ${\bf LHS}$ can be written by~\cite{Bertsekas_2009}
\begin{align}\label{Eq:relintLHS-original}
&{\rm relint}({\bf LHS})\coloneqq\nonumber\\
&\left\{{\bm\tau}\in{\bf LHS}\;\middle|\;\exists\;\epsilon>0\;\text{s.t.}\;\mathcal{B}({\bm\tau};\epsilon)\cap{\rm aff}({\bf LHS})\subseteq{\bf LHS}\right\},
\end{align}
where 
\begin{align}
% $
\mathcal{B}({\bm\tau};\epsilon)\coloneqq\left\{\{A_{a|x}\}_{a,x}\;\middle|\;\sum_{a,x}\norm{A_{a|x} - \tau_{a|x}}_1<\epsilon\right\}
% $
\end{align}
is an {\em open ball} centring at ${\bm\tau}$ with radius $\epsilon$.
Here, $A_{a|x}$'s are some Hermitian operators, and we use the metric 
% d({\bf A},{\bf B})\coloneqq
\CYthree{$\sum_{a,x}\norm{A_{a|x} - B_{a|x}}_1$} defined in Ref.~\cite{Hsieh2022}, which induces a metric topology. 
Also, the {\em affine hull} of the set ${\bf LHS}$ is given by
\begin{align}\label{Eq:affineLHS}
&{\rm aff}({\bf LHS})\coloneqq\nonumber\\
&\left\{\sum_{k=1}^{K}a_k{\bm\tau}_k\;\middle|\;K\in\mathbb{N}, a_k\in\mathbb{R}, {\bm\tau}_k\in{\bf LHS}, \sum_{k=1}^Ka_k=1\right\},
\end{align}
where $\alpha{\bm\sigma}+\beta{\bm\tau}\coloneqq\{\alpha\sigma_{a|x}+\beta\tau_{a|x}\}_{a,x}$.
Then we have:
\begin{lemma}\label{Eq:gamma_in_relint}
Let ${\bm\sigma}$ be a state assemblage with reduced state $\gamma$.
Suppose that $\gamma$ is full-rank and $p_{a|x}\coloneqq{\rm tr}(\sigma_{a|x})>0$ $\forall\,a,x$.
Then 
$
\{p_{a|x}\gamma\}_{a,x}\in{\rm relint}({\bf LHS}).
$
\end{lemma}
\begin{proof}
By Fact~\ref{Eq:LHSfact}, every \mbox{${\bm\kappa}\in{\bf LHS}$} can be expressed by
\mbox{$
\kappa_{a|x} = \sum_{i=1}^{|{\bf i}|}D(a|x,i)\eta_i\;\forall\,a,x
$} with some $\eta_i\ge0$ satisfying $\sum_i{\rm tr}(\eta_i)=1$, and there are $|{\bf i}|\coloneqq|{\bf a}|^{|{\bf x}|}$ many of deterministic probability distributions $D(a|x,i)$'s.
Now, we write
\begin{align}
p_{a|x} = (1-\epsilon)\widetilde{p}_{a|x} + \epsilon/|{\bf a}|,
\end{align}
where 
$
\widetilde{p}_{a|x}\coloneq(p_{a|x}-\epsilon/|{\bf a}|)/(1-\epsilon)
$
is a valid probability distribution with a small enough $\epsilon>0$ since we have $p_{a|x}>0$ $\forall\;a,x$.
Now, for $0<s<1$, we have
\begin{align}\label{Eq:Proof-no-go-001}
&p_{a|x}\gamma - s\kappa_{a|x} = (1-\epsilon)\widetilde{p}_{a|x}\gamma + \epsilon\gamma/|{\bf a}| - s\sum_{i=1}^{|{\bf i}|}D(a|x,i)\eta_i\nonumber\\
&= (1-\epsilon)\widetilde{p}_{a|x}\gamma + \sum_{i=1}^{|{\bf i}|}D(a|x,i)\left(\epsilon\gamma/|{\bf i}| - s\eta_i\right),
\end{align}
where, in the second line, we use the relation 
% \mbox{$
\begin{align}
\sum_{i=1}^{|{\bf i}|}D(a|x,i) = |{\bf a}|^{|{\bf x}|-1}= |{\bf i}|/|{\bf a}|\quad\forall\,a,x.
\end{align}
To prove this equality, note that $\sum_{i=1}^{|{\bf i}|}D(a|x,i)$ is the number of deterministic probability distributions that assign $a$ when the input is $x$.
For the given input $x$, the deterministic outcome $a$ has been assigned; for each other input, we have $|{\bf a}|$ many possible outcomes.
Hence, there are $|{\bf a}|^{|{\bf x}|-1}$ many such deterministic probability distributions.
Since $\gamma$ is full-rank, we have $\gamma\ge\mu_{\rm min}(\gamma)\id>0$, where $\mu_{\rm min}(\gamma)>0$ is the smallest eigenvalue of $\gamma$.
Then we choose
\begin{align}\label{Eq:the interval}
0<s<\min\left\{\epsilon;\frac{\epsilon\mu_{\rm min}(\gamma)}{|{\bf i}|\max_j\norm{\eta_j}_\infty}\right\}.
\end{align}
Note that $\epsilon\mu_{\rm min}(\gamma)/|{\bf i}|\max_j\norm{\eta_j}_\infty>0$ and this interval is non-vanishing.
For these value $s$ and every $i$, we obtain
\begin{align}
s\eta_i\le s\max_j\norm{\eta_j}_\infty\id\le\epsilon\mu_{\rm min}(\gamma)\id/|{\bf i}|\le\epsilon\gamma/|{\bf i}|.
\end{align}
Define the following operators for each $i$:
% $
\begin{align}
W_i\coloneqq\frac{\epsilon\gamma/|{\bf i}|-s\eta_i}{\epsilon-s},
\end{align}
which is well-defined since $\epsilon-s>0$ for every $s$ in the interval that we set [i.e., Eq.~\eqref{Eq:the interval}].
Then we have $W_i\ge0$ $\forall\,i$ and $\sum_{i=1}^{|{\bf i}|}{\rm tr}(W_i) = 1$.
By Fact~\ref{Eq:LHSfact}, we conclude that the state assemblage 
$\omega_{a|x} \coloneqq \sum_{i=1}^{|{\bf i}|}D(a|x,i)W_i$ is in ${\bf LHS}$.
Substituting everything back to Eq.~\eqref{Eq:Proof-no-go-001}, we obtain
\begin{align}\label{Eq:Proof-no-go-002}
\frac{p_{a|x}\gamma - s\kappa_{a|x}}{1-s} = \left(\frac{1-\epsilon}{1-s}\right)\widetilde{p}_{a|x}\gamma + \left(\frac{\epsilon-s}{1-s}\right)\omega_{a|x},
\end{align}
which is a convex mixture of two LHS state assemblages, meaning that $(p_{a|x}\gamma - s\kappa_{a|x})/(1-s)\in{\bf LHS}$ due to the convexity of ${\bf LHS}$.
Finally, since ${\bf LHS}$ is convex, its relative interior can be written by~\cite{Bertsekas_2009}
\begin{align}\label{Eq:relintLHS}
{\rm relint}({\bf LHS})=\{{\bm\tau}\;|\;&\forall\,{\bm\kappa}\in{\bf LHS},\exists\,0<s<1\nonumber\\
&\;\text{s.t.}\;({\bm\tau} - s{\bm\kappa})/(1-s)\in{\bf LHS}\}.
\end{align}
The proof is thus concluded since our argument works for every ${\bm\kappa}\in{\bf LHS}$.
\end{proof}

\subsubsection*{Proof of \CYnew{Result~\ref{Result:Strong-no-go}}}
\CYnew{To prove Result~\ref{Result:Strong-no-go} in the main text,} we first point out the following fact:
\begin{fact}\label{fact:affineLHS}
For every state assemblage ${\bm\sigma}$, there exist two state assemblages ${\bm\tau},{\bm\omega}$ in ${\bf LHS}$ and $0<p<1$ such that
\begin{align}
{\bm\sigma} = \frac{{\bm\tau} + (p-1){\bm\omega}}{p}.
\end{align}
In other words, every state assemblage is in ${\rm aff}({\bf LHS})$.
\end{fact}
\begin{proof}
First, for every ${\bm\sigma}$, there exists some bipartite state $\rho_{AB}$ and positive operator-valued measures $\{E_{a|x}\}_{a,x}$ (i.e., $\sum_aE_{a|x} = \id$ $\forall\,x$ and $E_{a|x}\ge0$ $\forall\,a,x$; namely, for each $x$, the set $\{E_{a|x}\}_{a}$ describes a measurement) achieving
$
\sigma_{a|x} = {\rm tr}_A\left[\rho_{AB}(E_{a|x}\otimes\id_B)\right].
$
Since there exists some $0<p<1$ such that
$
p\rho_{AB}+(1-p)\id_{AB}/d^2
$
is a separable state~\cite{Gurvits2002},
we obtain
\mbox{$
p{\bm\sigma} + (1-p){\bm\omega}\in{\bf LHS},
$}
where $\omega_{a|x}\coloneqq{\rm tr}_A\left[\id_{AB}/d^2(E_{a|x}\otimes\id_B)\right]$ is again in ${\bf LHS}$.
The result follows by using Eq.~\eqref{Eq:affineLHS}.
\end{proof}

With the above observation, we are now ready to prove the following result, \CYnew{which restates Result~\ref{Result:Strong-no-go} in the main text:}

\begin{result}
Let ${\bm\sigma}$ be a state assemblage with reduced state $\gamma$.
Suppose that $\gamma$ is full-rank, and $p_{a|x}\coloneqq{\rm tr}(\sigma_{a|x})>0$ for every $a,x$.
Let $\mathfrak{E} = \{\mathcal{N}_t\}_{t=0}^\infty$ be an evolution consisting of positive trace-preserving linear maps.
If $\mathfrak{E}$ thermalises the system to $\gamma$, there must exists a finite $t_*$ such that
\begin{align}
\mathcal{N}_t({\bm\sigma})\in{\bf LHS}\quad\forall\;t>t_*.
\end{align} 
Namely, signatures of incompatibility and steering must vanish at a finite time.
\end{result}
\begin{proof}
By Lemma~\ref{Eq:gamma_in_relint} and Eq.~\eqref{Eq:relintLHS-original}, there is $\epsilon_*>0$ such that
\begin{align}\label{Eq:proof-no-go-001}
\mathcal{B}\left(\{p_{a|x}\gamma\}_{a,x};\epsilon_*\right)\cap{\rm aff}({\bf LHS})\subseteq{\bf LHS}.
\end{align}
On the other hand, we have 
\mbox{$
\lim_{t\to\infty}\norm{\mathcal{N}_t(\cdot) - \gamma{\rm tr}(\cdot)}_\diamond = 0.
$}
This means that
\mbox{$
\lim_{t\to\infty}\norm{\mathcal{N}_t(\sigma_{a|x}) - p_{a|x}\gamma}_1 = 0\;\forall\;a,x.
$}
Hence, there exists a time point $t_*<\infty$ such that
\mbox{$
\sum_{a,x}\norm{\mathcal{N}_t(\sigma_{a|x}) - p_{a|x}\gamma}_1 < \epsilon_*\;\forall\;t>t_*
$.}
In other words,
\begin{align}\label{Eq:proof-no-go-002}
\mathcal{N}_t({\bm\sigma})\in\mathcal{B}\left(\{p_{a|x}\gamma\}_{a,x};\epsilon_*\right)\;\forall\;t>t_*.
\end{align}
Since each $\mathcal{N}_t$ is positive trace-preserving linear map, $\{\mathcal{N}_t(\sigma_{a|x})\}_{a,x}$ is a valid state assemblage and hence
 is in ${\rm aff}({\bf LHS})$ by Fact~\ref{fact:affineLHS}.
Finally, using \CYnew{Eq.~\eqref{Eq:proof-no-go-001},} we conclude that, for every $t>t_*$,
\begin{align}
\mathcal{N}_t({\bm\sigma})\in\mathcal{B}\left(\{p_{a|x}\gamma\}_{a,x};\epsilon_*\right)\cap{\rm aff}({\bf LHS})\subseteq{\bf LHS}.
\end{align}
Namely, $\mathcal{N}_t({\bm\sigma})\in{\bf LHS}\;\forall\,t>t_*$, as desired.
\end{proof}

\subsection*{\CYnew{Supplemental Material VIII:} Conic Programming and \CYnew{Result~\ref{Result:WorkExtraction}}}\label{App:Proof-Result:WorkExtraction}

\subsubsection*{Conic Programming}
Our proof relies on a specific form of the {\em conic programming}. 
We refer the readers to, e.g., Refs.~\cite{Uola2019PRL,Uola2020PRL,Takagi2019,Hsieh2023-2} for further details, and here we only recap the necessary ingredients. 
Since now, we use the notation ${\bf A}\coloneqq\bigoplus_{n=0}^{N-1}A_n$ for a set of Hermitian operators $\{A_n\}_{n=0}^{N-1}$.
In the form relevant to this work, we call the following optimisation the {\em primal problem}:
\begin{equation}\label{Eq:CPPrimal}
\begin{split}
\min_{{\bf X}}\quad&\sum_{n=0}^{N-1}{\rm tr}(A_nX_n)\\
{\rm s.t.}\quad& {\bf X}\in\C;\;\mathfrak{L}({\bf X})= {\bf B}.
\end{split}
\end{equation}
Here, $\sum_{n=0}^{N-1}{\rm tr}(A_nX_n)$ is an inner product between Hermitian ${\bf A},{\bf X}$, $\mathfrak{L}$ is a linear map, and ${\bf B} = \bigoplus_{n=0}^{N-1}B_n$ is a constant Hermitian operator.
Furthermore, $\C$ is a convex and closed {\em cone} (i.e., if ${\bf X}\in\C$, then $\alpha {\bf X}\in\C$ $\forall\;\alpha\ge0$; see, e.g., Appendix B in Ref.~\cite{Takagi2019}).
The {\em dual problem} of Eq.~\eqref{Eq:CPPrimal} is given by~\cite{Takagi2019}
\begin{equation}\label{Eq:CPDual}
\begin{split}
\max_{{\bf Y}}\quad&\sum_{n=0}^{N-1}{\rm tr}(B_nY_n)\\
{\rm s.t.}\quad&Y_n=Y_n^\dagger\;\forall\;n;\\
&\sum_{n=0}^{N-1}{\rm tr}\left[Y_n\mathfrak{L}({\bf X})_n\right]\le\sum_{n=0}^{N-1}{\rm tr}(A_nX_n)\;\forall\;{\bf X}\in\C,
\end{split}
\end{equation}
where $\mathfrak{L}({\bf X})_n$ is the $n$-th element of $\mathfrak{L}({\bf X})$.
Equation~\eqref{Eq:CPDual} always lower bounds Eq.~\eqref{Eq:CPPrimal}.
When they coincide, it is called the {\em strong duality}~\cite{Boyd-Book}.
This happens if there exists some ${\bf X}\in{\rm relint}\left(\C\right)$ (the relative interior of $\mathcal{C}$) such that $\mathfrak{L}({\bf X}) = {\bf B}$ --- the so-called {\em Slater's condition}~\cite{Boyd-Book} (see also Ref.~\cite{Takagi2019}).

\subsubsection*{Proof of \CYnew{Result~\ref{Result:WorkExtraction}}}
\CYtwo{Our proof strategy is as follows: First, we prove the case for $\mathbfcal{H}_{\eta=1}$; that is, when we set $\eta=1$.
After that, we prove the general case for $\mathbfcal{H}_{\eta}$ for every $0<\eta<\infty$.
The following is thus the proof for the case with $\mathbfcal{H}_{\eta=1}$.}
\begin{proof}
({\em Computing the dual problem}) 
First of all, note that we always assume $\gamma$ is full-rank and $p_{a|x}\coloneqq{\rm tr}(\sigma_{a|x})>0$ for every $a,x$.
These assumptions will play a crucial role in our proof.
To start with, consider a given {\em steerable} state assemblage ${\bm\sigma}\notin{\bf LHS}$.
To compute its ${\rm SR}_\gamma$, let us rewrite \CYnew{Eq.~\eqref{Eq:SRgamma} in the main text} as
\begin{equation}\label{Eq: Eq4}
\begin{split}
2^{{\rm SR}_\gamma({\bm\sigma})}\coloneqq\min_{\SRvar,{\bm\tau}}\quad&\SRvar\\
{\rm s.t.}\quad& \SRvar\ge1;\;{\bm\tau}\in{\bf LHS};\\
&\frac{\sigma_{a|x} + (\SRvar-1)\gamma p_{a|x}}{\SRvar} = \tau_{a|x}\;\forall\;a,x.
\end{split}
\end{equation}
This minimisation is lower bounded by the following one by extending the minimisation range to $\SRvar\ge0$:
\begin{equation}\label{Eq: min with lambda ge0}
\begin{split}
\min_{\SRvar,{\bm\tau}}\quad&\SRvar\\
{\rm s.t.}\quad& \SRvar\ge0;\;{\bm\tau}\in{\bf LHS};\\
&\frac{\sigma_{a|x} + (\SRvar-1)\gamma p_{a|x}}{\SRvar} = \tau_{a|x}\;\forall\;a,x.
\end{split}
\end{equation}
Importantly, since ${\bm\sigma}\notin{\bf LHS}$, every feasible $\SRvar$ in Eq.~\eqref{Eq: min with lambda ge0} must satisfy $\SRvar>1$. 
To see this, if $\sigma_{a|x} = \SRvar\tau_{a|x} + (1-\SRvar)\gamma p_{a|x}$ for some $0\le\SRvar\le1$ and ${\bm\tau}\in{\bf LHS}$, then it is a convex mixture of two LHS state assemblages, meaning that ${\bm\sigma}\in{\bf LHS}$, a contradiction.
Hence,  Eq.~\eqref{Eq: min with lambda ge0} is, in fact, equal to Eq.~\eqref{Eq: Eq4} whenever ${\bm\sigma}\notin{\bf LHS}$.
Now, we further observe that every feasible ${\bm\tau}$ of Eq.~\eqref{Eq: min with lambda ge0} must satisfy (note that feasible $\SRvar$ must be non-zero)
\begin{align}
\sum_a\tau_{a|x} &= \gamma = \sum_a\sigma_{a|x}\quad\forall\,x;\nonumber\\
{\rm tr}(\tau_{a|x}) &= p_{a|x}\coloneqq{\rm tr}(\sigma_{a|x})\quad\forall\;a,x.\label{Eq:classical statistic condition}
\end{align}
Hence, we can rewrite $2^{{\rm SR}_\gamma({\bm\sigma})}$ as the following minimisation:
\begin{equation}
\begin{split}
2^{{\rm SR}_\gamma({\bm\sigma})}=\min_{\SRvar,{\bm\tau}}\quad&\SRvar\\
{\rm s.t.}\quad& \SRvar\ge0;\;{\bm\tau}\in{\bf LHS}({\bm\sigma});\\
&\sigma_{a|x} - \gamma p_{a|x} = \SRvar\left[\tau_{a|x} - \gamma {\rm tr}(\tau_{a|x})\right]\;\forall\;a,x,
\end{split}
\end{equation}
where, recall from the main text, that
\begin{align}\label{Eq: LHS sigma def}
{\bf LHS}({\bm\sigma})\coloneqq\{{\bm\tau}\in{\bf LHS}\,|\,{\rm tr}(\sigma_{a|x})={\rm tr}(\tau_{a|x})\;\forall\,a,x\}
\end{align}
is the set of all LHS state assemblages with the same classical statistics with ${\bm\sigma}$.
Now, we define the variable 
\begin{align}
{\bf V}\coloneqq\bigoplus_{a,x}V_{a|x}
\end{align}
with 
\begin{align}
V_{a|x}\coloneqq\SRvar\tau_{a|x}\quad\forall\,a,x.
\end{align}
We also define the cone 
\begin{align}
\C_{\rm LHS}\coloneqq\left\{{\bf V} = \bigoplus_{a,x}\SRvar\tau_{a|x}\;\middle|\;\SRvar\ge0, {\bm\tau}\in{\bf LHS}({\bm\sigma})\right\},
\end{align}
which can be checked to be a convex and closed cone (i.e., it is a convex and closed set such that if ${\bf V}\in\C_{\rm LHS}$, then \mbox{$\alpha{\bf V}\in\C_{\rm LHS}$} $\forall\,\alpha\ge0$).
See, e.g., Appendix B in Ref.~\cite{Takagi2019} and Ref.~\cite{Hsieh2023-2} for a pedagogical explanation.
For \mbox{${\bf V}\in\C_{\rm LHS}$}, we have $\sum_{a}{\rm tr}(V_{a|x}) = \SRvar$ for every fixed $x$.
This means that, by summing over $x=0,...,|{\bf x}|-1$, we obtain
\begin{align}
\SRvar = \frac{1}{|{\bf x}|}\sum_{a,x}{\rm tr}(V_{a|x}).
\end{align}
Hence, for \mbox{${\bf V}\in\C_{\rm LHS}$}, we can write $2^{{\rm SR}_\gamma({\bm\sigma})}$ as
\begin{equation}
\begin{split}
\min_{{\bf V}}\quad&\frac{1}{|{\bf x}|}\sum_{a,x}{\rm tr}(V_{a|x})\\
{\rm s.t.}\quad& {\bf V}\in\C_{\rm LHS};\\
&\bigoplus_{a,x}\left(\sigma_{a|x} - \gamma p_{a|x}\right) = \bigoplus_{a,x}\left[V_{a|x} - \gamma {\rm tr}(V_{a|x})\right],
\end{split}
\end{equation}
which is in the standard form of conic programming as given in Eq.~\eqref{Eq:CPPrimal} by substituting
\begin{align}\label{Eq:CP variables conditions}
{\bf X}&\coloneqq{\bf V};\nonumber\\
\C&\coloneqq\C_{\rm LHS};\nonumber\\
{\bf A}&\coloneqq\bigoplus_{a,x}\left(\frac{\id}{|{\bf x}|}\right);\nonumber\\
{\bf B}&\coloneqq\bigoplus_{a,x}\left(\sigma_{a|x} - \gamma p_{a|x}\right);\nonumber\\
\mathfrak{L}({\bf V})&\coloneqq\bigoplus_{a,x}\left[V_{a|x} - \gamma {\rm tr}(V_{a|x})\right].
\end{align}
Using Eq.~\eqref{Eq:CPDual}, its dual problem is thus given by
\begin{equation}\label{Eq:Dual computation 0001}
\begin{split}
\max_{{\bf\Ydual}}\quad&\sum_{a,x}{\rm tr}\left[\Ydual_{a|x}\left(\sigma_{a|x} - \gamma p_{a|x}\right)\right]\\
{\rm s.t.}\quad& \Ydual_{a|x}=\Ydual_{a|x}^\dagger\;\forall\;a,x;\\
&\max_{{\bm\tau}'\in{\bf LHS}({\bm\sigma})}\sum_{a,x}{\rm tr}\left[\Ydual_{a|x}\left(\tau'_{a|x}-\gamma {\rm tr}(\tau'_{a|x})\right)\right]\le1,
\end{split}
\end{equation}
where $\SRvar$'s are cancelled out (and note that the constraint associated with $\SRvar=0$ always holds).
Note that the above constraint contains a maximisation over all state assemblages from ${\bf LHS}({\bm\sigma})$ (denoted by ${\bm\tau}'$).
They form a set that is, in general, greater than the set of all possible primal variables ${\bm\tau}$ [which must satisfy Eq.~\eqref{Eq:classical statistic condition}].
We thus use a different notation to avoid confusion.
Now, for a state $\rho$ and Hamiltonian $H$, from Ref.~\cite{Hsieh2023IP}, we have
(here, $d$ is the system dimension)
\begin{align}\label{Eq:Wdiff formula}
\Wdiff(\rho,H) = {\rm tr}(H\rho) + k_BT\ln{\rm tr}\left(e^{-\frac{H}{k_BT}}\right) - k_BT\ln d,
\end{align}
This implies the following useful relation for every two states $\rho,\eta$ (in the same system) and Hamiltonian $H$:
\begin{align}
\Wdiff(\rho,H) - \Wdiff(\eta,H) = {\rm tr}\left[H(\rho-\eta)\right].
\end{align}
Consequently, together with the definition of $\WdiffSA$, we obtain the following relation for {\em any} state assemblage \mbox{${\bm\omega}=\{\omega_{a|x}\}_{a,x}$}:
\begin{align}\label{Eq:useful2}
|{\bf x}|\times\WdiffSA({\bm\omega},{\bf H})&=\sum_{a,x}{\rm tr}(\omega_{a|x})\left[\Wdiff(\hat{\omega}_{a|x},H_{a|x}) - \Wdiff(\gamma,H_{a|x})\right]\nonumber\\
&=\sum_{a,x}{\rm tr}\left[H_{a|x}\left(\omega_{a|x}-{\rm tr}(\omega_{a|x})\gamma\right)\right],
\end{align}
where recall that $\hat{\omega}_{a|x}\coloneqq\omega_{a|x}/{\rm tr}(\omega_{a|x})$ is the normalised state.
Then, by defining the Hamiltonian
\begin{align}
H_{a|x}\coloneqq (k_BT)|{\bf x}|\times\Ydual_{a|x}\quad\forall\,a,x,
\end{align}
and using Eq.~\eqref{Eq:useful2},
we can write the dual problem Eq.~\eqref{Eq:Dual computation 0001} as
\begin{equation}\label{Eq:dual_0001}
\begin{split}
\max_{{\bf H}}\quad&\frac{\WdiffSA({\bm\sigma},{\bf H})}{k_BT}\\
{\rm s.t.}\quad& H_{a|x}=H_{a|x}^\dagger\;\forall\;a,x;\;\;\max_{{\bm\tau}'\in{\bf LHS}({\bm\sigma})}\WdiffSA({\bm\tau}',{\bf H})\le{k_BT},
\end{split}
\end{equation}
which is maximising over all possible Hamiltonians ${\bf H}$.
Note that each $H_{a|x}$ now carries the dimensionality of energy, and we always assume that $T$ is fixed, finite, and strictly positive.
Also, in general, $\gamma$ is {\em not} a thermal state subject to $H_{a|x}$ in this optimisation.
One should just view $\gamma$ as a quantum state here.

({\em Checking Slater's condition}) When the strong duality holds, Eq.~\eqref{Eq:dual_0001} will equal to $2^{{\rm SR}_\gamma({\bm\sigma})}$.
Hence, let us check the Slater's condition now.
First, for $\SRvar\ge1$, define \mbox{${\bm\sigma}^{(\SRvar)}\coloneqq\left\{\sigma_{a|x}^{(\SRvar)}\right\}_{a,x}$} with
\begin{align}\label{Eq:sigma^lambda_def}
\sigma_{a|x}^{(\SRvar)}\coloneqq\frac{1}{\SRvar}\sigma_{a|x}+\left(1-\frac{1}{\SRvar}\right)p_{a|x}\gamma,
\end{align}
which is a valid state assemblage since, importantly, $\SRvar\ge1$.
Also, it satisfies
\begin{align}\label{Eq:sigma_lambda_estimate}
\sum_{a,x}\norm{\sigma_{a|x}^{(\SRvar)} - p_{a|x}\gamma}_1 = \frac{1}{\SRvar}\sum_{a,x}\norm{\sigma_{a|x} - p_{a|x}\gamma}_1.
\end{align}
This can be as small as we want by considering a large enough (while still finite) $\SRvar$.
Using the same proof of Lemma~\ref{Eq:gamma_in_relint}, one can show that \mbox{$\{p_{a|x}\gamma\}_{a,x}\in{\rm relint}\left({\bf LHS}({\bm\sigma})\right)$}, where
\begin{align}\label{Eq:relintLHS-sigma-original}
{\rm relint}({\bf LHS}(\bm\sigma))\coloneqq&\{{\bm\tau}\in{\bf LHS}(\bm\sigma)\;|\;\exists\;\epsilon>0\;\text{s.t.}\;\nonumber\\
&\quad\mathcal{B}({\bm\tau};\epsilon)\cap{\rm aff}({\bf LHS}(\bm\sigma))\subseteq{\bf LHS}(\bm\sigma)\}
\end{align}
is the relative interior of ${\bf LHS}(\bm\sigma)$ [see also Eq.~\eqref{Eq:relintLHS-original}].
Then,
there exists \mbox{$\epsilon_*>0$} such that
% \mbox{$
\begin{align}
\mathcal{B}\left(\{p_{a|x}\gamma\}_{a,x};\epsilon_*\right)\cap{\rm aff}({\bf LHS}({\bm\sigma}))\subseteq{\bf LHS}({\bm\sigma}).
\end{align}
% $}
Here shows the importance of assuming $\gamma$ is full-rank and \mbox{$p_{a|x}\coloneqq{\rm tr}(\sigma_{a|x})>0$} $\forall\,a,x$ --- so that the proof of Lemma~\ref{Eq:gamma_in_relint} can apply.
By choosing a finite and sufficiently large $\SRvar_*>1$ and use Eq.~\eqref{Eq:sigma_lambda_estimate}, we are able to make sure
\mbox{$
{\bm\sigma}^{(\SRvar_*)}\in\mathcal{B}\left(\{p_{a|x}\gamma\}_{a,x};\epsilon_*\right).
$}
This means that we can choose an even smaller \mbox{$\delta_*>0$} achieving 
\begin{align}
\mathcal{B}\left({\bm\sigma}^{(\SRvar_*)};\delta_*\right)\subseteq\mathcal{B}\left(\{p_{a|x}\gamma\}_{a,x};\epsilon_*\right).
\end{align}
This implies that
$
\mathcal{B}\left({\bm\sigma}^{(\SRvar_*)};\delta_*\right)\cap{\rm aff}({\bf LHS}({\bm\sigma}))\subseteq{\bf LHS}({\bm\sigma}).
$
Also, we have ${\bm\sigma}^{(\SRvar_*)}\in{\rm aff}({\bf LHS}({\bm\sigma}))$ [see Fact~\ref{lemma: LHS sigma relint} and Eq.~\eqref{Eq:sigma lambda in aff LHS sigma} for details].
This thus means \mbox{${\bm\sigma}^{(\SRvar_*)}\in{\bf LHS}({\bm\sigma})$}.
Using the definition Eq.~\eqref{Eq:relintLHS-sigma-original}, we conclude that 
\mbox{${\bm\sigma}^{(\SRvar_*)}\in{\rm relint}({\bf LHS}({\bm\sigma}))$}.
Now, define
\begin{align}
{\bf V}_*\coloneqq\bigoplus_{a,x}\SRvar_*{\sigma}_{a|x}^{(\SRvar_*)}.
\end{align}
Using Lemma A.3 in \CYnew{the first arXiv version of} Ref.~\cite{Hsieh2023-2}, Eqs.~\eqref{Eq:CP variables conditions} and~\eqref{Eq:sigma^lambda_def}, one can then check that
\begin{align}
{\bf V}_*\in{\rm relint}(\mathcal{C}_{\rm LHS})\quad\&\quad\mathfrak{L}({\bf V}_*) = {\bf B}.
\end{align}
Hence, Slater's conditions are satisfied, and the strong duality holds, and Eq.~\eqref{Eq:dual_0001} is the same as $2^{{\rm SR}_\gamma({\bm\sigma})}$.

({\em Proving the upper bound})
Using the strong duality (which is true whenver ${\bm\sigma}\notin{\bf LHS}$), Eq.~\eqref{Eq:dual_0001} can be rewritten as
\begin{equation}
\begin{split}
2^{{\rm SR}_\gamma({\bm\sigma})}=\max_{{\bf H}}\quad&\frac{\WdiffSA({\bm\sigma},{\bf H})}{k_BT}\\
{\rm s.t.}\quad& H_{a|x}=H_{a|x}^\dagger\;\forall\;a,x;\\
&\max_{{\bm\tau}'\in{\bf LHS}({\bm\sigma})}\WdiffSA({\bm\tau}',{\bf H})\le{k_BT}\le\WdiffSA({\bm\sigma},{\bf H}).
\end{split}
\end{equation}
This is because by definition we always have \mbox{$2^{{\rm SR}_\gamma({\bm\sigma})}\ge1$;} hence, putting the condition $\WdiffSA({\bm\sigma},{\bf H})\ge{k_BT}$ should output the same optimal value.
Using the constraint $\max_{{\bm\tau}'\in{\bf LHS}({\bm\sigma})}\WdiffSA({\bm\tau}',{\bf H})\le k_BT$, we obtain the desired bound for every ${\bf H}$ satisfying $\WdiffSA({\bm\sigma},{\bf H})\ge k_BT$, which is written as ${\bf H}\in\mathbfcal{H}_{\eta=1}({\bm\sigma})$ in the main text:
\begin{align}\label{Eq:lemma:upperbound}
2^{{\rm SR}_\gamma({\bm\sigma})}\le\max_{{\bf H}\in\CYtwo{\mathbfcal{H}_{\eta=1}}({\bm\sigma})}\frac{\WdiffSA({\bm\sigma},{\bf H})}{\max_{{\bm\tau}'\in{\bf LHS}({\bm\sigma})}\WdiffSA({\bm\tau}',{\bf H})}.
\end{align}

({\em Achieving the upper bound})
Finally, we show that the upper bound in Eq.~\eqref{Eq:lemma:upperbound} can be achieved.
First of all, note that two finite-dimensional Hermitian operators $P,Q$ are the same, i.e., $P=Q$, if and only if ${\rm tr}(PA) = {\rm tr}(QA)$ for every Hermitian operator $A$.
Hence, for any ${\bm\sigma}$ with the reduced state $\gamma$, we can rewrite the primal problem of $2^{{\rm SR}_\gamma({\bm\sigma})}$ \CYnew{[i.e., Eq.~\eqref{Eq:SRgamma} in the main text]} as
\begin{equation}
\begin{split}\label{Eq:minimisation0000}
\min_{\SRvar,{\bm\tau}}\quad&\SRvar\\
{\rm s.t.}\quad&\SRvar\ge1;\;\;{\bm\tau}\in{\bf LHS}({\bm\sigma});\\ 
&{\rm tr}\left[H_{a|x}\left(\sigma_{a|x}-p_{a|x}\gamma\right)\right] = \SRvar\times{\rm tr}\left[H_{a|x}\left(\tau_{a|x}-p_{a|x}\gamma\right)\right]\\
&\forall\,a,x,H_{a|x}=H_{a|x}^\dagger.
\end{split}
\end{equation}
Recall that the {\em primal} variable ${\bm\tau}$ satisfies \mbox{$p_{a|x} = {\rm tr}(\tau_{a|x})$} [see Eq.~\eqref{Eq:classical statistic condition}].
Using Eq.~\eqref{Eq:useful2}, for {\em every fixed} collection of Hamiltonians ${\bf H} = \{H_{a|x}\}_{a,x}$, we have that
\begin{align}\label{Eq:minimisation001}
&2^{{\rm SR}_\gamma({\bm\sigma})}\ge\nonumber\\
&\min\left\{\SRvar\ge1\,\middle|\,\exists\;{\bm\tau}\in{\bf LHS}({\bm\sigma})\,\text{s.t.}\,\WdiffSA({\bm\sigma},{\bf H}) = \SRvar\WdiffSA({\bm\tau},{\bf H})\right\}.
\end{align}
This is because every feasible solution $(\SRvar,{\bm\tau})$ of Eq.~\eqref{Eq:minimisation0000} will achieve $\WdiffSA({\bm\sigma},{\bf H}) = \SRvar\WdiffSA({\bm\tau},{\bf H})$ for the {\em given} ${\bf H}$ according to Eq.~\eqref{Eq:useful2}.
Hence, Eq.~\eqref{Eq:minimisation0000} is sub-optimal to the minimisation in Eq.~\eqref{Eq:minimisation001}.
Now, let us consider an arbitrary ${\bf H}$ from \CYtwo{$\mathbfcal{H}_{\eta=1}({\bm\sigma})$}. 
Then, for this ${\bf H}$, any feasible solution $(\SRvar,{\bm\tau})$ to the minimisation in Eq.~\eqref{Eq:minimisation001} must satisfy
\begin{align}
{0<k_BT}\le\WdiffSA({\bm\sigma},{\bf H}) = \SRvar\WdiffSA({\bm\tau},{\bf H})\le\SRvar\max_{{\bm\tau}'\in{\bf LHS}({\bm\sigma})}\WdiffSA({\bm\tau}',{\bf H}).
\end{align}
Note that $T$ is always considered finite and strictly positive.
Also, again, the above maximisation is taken over all possible state assemblages from ${\bf LHS}({\bm\sigma})$ (denoted by ${\bm\tau}'$), which form a set that is larger than the one of all possible primal variables ${\bm\tau}$ [that satisfies Eq.~\eqref{Eq:classical statistic condition}].
This means that, for {\em every} feasible $\SRvar$ to Eq.~\eqref{Eq:minimisation001} with an \CYtwo{${\bf H}\in\mathbfcal{H}_{\eta=1}({\bm\sigma})$,} one can divide $\max_{{\bm\tau}'\in{\bf LHS}({\bm\sigma})}\WdiffSA({\bm\tau}',{\bf H})$ (which is positive) and obtain
\begin{align}\label{Eq:lambda lower bound}
\SRvar\ge\frac{\WdiffSA({\bm\sigma},{\bf H})}{\max_{{\bm\tau}'\in{\bf LHS}({\bm\sigma})}\WdiffSA({\bm\tau}',{\bf H})}.
\end{align}
Using Eq.~\eqref{Eq:minimisation001} and maximising over \CYtwo{${\bf H}\in\mathbfcal{H}_{\eta=1}({\bm\sigma})$,} we obtain
\begin{align}
2^{{\rm SR}_\gamma({\bm\sigma})}
&\ge\max_{{\bf H}\in\CYtwo{\mathbfcal{H}_{\eta=1}}({\bm\sigma})}\min\Big\{\SRvar\ge1\,\Big|\,\exists\;{\bm\tau}\in{\bf LHS}({\bm\sigma})\;\text{s.t.}\nonumber\\
&\quad\quad\quad\quad\quad\quad\quad\quad\quad\quad\;\;\;\WdiffSA({\bm\sigma},{\bf H}) = \SRvar\WdiffSA({\bm\tau},{\bf H})\Big\}\nonumber\\
&\ge\max_{{\bf H}\in\CYtwo{\mathbfcal{H}_{\eta=1}}({\bm\sigma})}\frac{\WdiffSA({\bm\sigma},{\bf H})}{\max_{{\bm\tau}'\in{\bf LHS}({\bm\sigma})}\WdiffSA({\bm\tau}',{\bf H})},
\end{align}
where we have used Eq.~\eqref{Eq:lambda lower bound}.
This means that the upper bound in Eq.~\eqref{Eq:lemma:upperbound} can be achieved for every given ${\bm\sigma}\notin{\bf LHS}$.
\end{proof}

\CYtwo{The above proof completes the case for $\mathbfcal{H}_{\eta=1}$.
Now, we show the case for every $0<\eta<\infty$:}

\CYtwo{
\begin{proof}
What we just proved is
\begin{align}\label{Eq: general case 001}
&2^{{\rm SR}_\gamma({\bm\sigma})}
=\max_{{\bf H}\in\mathbfcal{H}_{\eta=1}({\bm\sigma})}\frac{\WdiffSA({\bm\sigma},{\bf H})}{\max_{{\bm\tau}'\in{\bf LHS}({\bm\sigma})}\WdiffSA({\bm\tau}',{\bf H})}\nonumber\\
&=\max_{{\bf H}\in\mathbfcal{H}_{\eta=1}({\bm\sigma})}\frac{\sum_{a,x}{\rm tr}\left[H_{a|x}\left(\sigma_{a|x}-{\rm tr}(\sigma_{a|x})\gamma\right)\right]}{\max_{{\bm\tau}'\in{\bf LHS}({\bm\sigma})}\sum_{a,x}{\rm tr}\left[H_{a|x}\left(\tau'_{a|x}-{\rm tr}(\tau'_{a|x})\gamma\right)\right]}\nonumber\\
&=\max_{{\bf H}\in\mathbfcal{H}_{\eta=1}({\bm\sigma})}\frac{\sum_{a,x}{\rm tr}\left[H_{a|x}\left(\sigma_{a|x}-p_{a|x}\gamma\right)\right]}{\max_{{\bm\tau}'\in{\bf LHS}({\bm\sigma})}\sum_{a,x}{\rm tr}\left[H_{a|x}\left(\tau'_{a|x}-p_{a|x}\gamma\right)\right]},
\end{align}
where we have used Eq.~\eqref{Eq:useful2} and the fact that ${\rm tr}(\tau'_{a|x}) = {\rm tr}(\sigma_{a|x}) = p_{a|x}$ $\forall\,a,x$ 
 for ${\bm\tau}'\in{\bf LHS}({\bm\sigma})$ [see also Eqs.~\eqref{Eq:classical statistic condition} and~\eqref{Eq: LHS sigma def}].
Now, consider a parameter $0<\eta<\infty$.
Using Eq.~\eqref{Eq:useful2} again, we have ${\bf H}\in\mathbfcal{H}_{\eta=1}({\bm\sigma})$ if and only if
\begin{align}
\sum_{a,x}{\rm tr}\left[H_{a|x}\left(\sigma_{a|x}-p_{a|x}\gamma\right)\right] = |{\bf x}|\times\WdiffSA({\bm\sigma},{\bf H})\ge k_BT|{\bf x}|,
\end{align}
which is true if and only if $\{\eta\times H_{a|x}\}_{a,x}\in\mathbfcal{H}_\eta({\bm\sigma})$.
Hence, by multiplying both the numerator and the denominator by $\eta$ and defining a new variable \mbox{${\bf H}^{(\eta)}\coloneqq\{\eta\times H_{a|x}\}_{a,x}$}, one is able to rewrite Eq.~\eqref{Eq: general case 001} into
\begin{align}
&2^{{\rm SR}_\gamma({\bm\sigma})}
=\max_{{\bf H}\in\mathbfcal{H}_{\eta=1}({\bm\sigma})}\frac{\WdiffSA({\bm\sigma},{\bf H})}{\max_{{\bm\tau}'\in{\bf LHS}({\bm\sigma})}\WdiffSA({\bm\tau}',{\bf H})}\nonumber\\
&=\max_{{\bf H}^{(\eta)}\in\mathbfcal{H}_{\eta}({\bm\sigma})}\frac{\sum_{a,x}{\rm tr}\left[H_{a|x}^{(\eta)}\left(\sigma_{a|x}-p_{a|x}\gamma\right)\right]}{\max_{{\bm\tau}'\in{\bf LHS}({\bm\sigma})}\sum_{a,x}{\rm tr}\left[H_{a|x}^{(\eta)}\left(\tau'_{a|x}-p_{a|x}\gamma\right)\right]}\nonumber\\
&=\max_{{\bf H}^{(\eta)}\in\mathbfcal{H}_{\eta}({\bm\sigma})}\frac{\WdiffSA({\bm\sigma},{\bf H}^{(\eta)})}{\max_{{\bm\tau}'\in{\bf LHS}({\bm\sigma})}\WdiffSA({\bm\tau}',{\bf H}^{(\eta)})},
\end{align}
where we have used Eq.~\eqref{Eq:useful2} again.
This thus completes the proof of Result~\ref{Result:WorkExtraction} in the main text.
\end{proof}
}

It turns out that we can use SDP to efficiently find Hamiltonians to certify quantum signatures:
\begin{proposition}\label{coro:SDP-WorkExtraction}
Let ${\bm\sigma}$ be a state assemblage with reduced state $\gamma$.
Then ${\rm SR}_\gamma({\bm\sigma})>0$ if and only if there is ${\bf H}$ 
such that
\begin{align}
\WdiffSA({\bm\sigma},{\bf H})>\max_{{\bm\tau}\in{\bf LHS}({\bm\sigma})}\WdiffSA({\bm\tau},{\bf H}).
\end{align}
Moreover, running SDP can explicitly find this ${\bf H}$.
\end{proposition}
\begin{proof}
First, since ${\rm SR}_\gamma({\bm\sigma})>0$ if and only if ${\bm\sigma}\notin{\bf LHS}$, which is true if and only if ${\bm\sigma}\notin{\bf LHS}({\bm\sigma})$, it suffices to show that ${\rm SR}_\gamma({\bm\sigma})>0$ implies the existence of some ${\bf H}$ that can achieve the desired strict inequality.
When we have ${\rm SR}_\gamma({\bm\sigma})>0$ (i.e., $2^{-{\rm SR}_\gamma({\bm\sigma})}<1$), the dual SDP of Eq.~\eqref{Eq:SRgamma} in the main text, that is, Eq.~\eqref{Eq:SRgammaSDPdual}, implies that there exists some Hermitian operators $\widetilde{\Ydual}_{a|x} = \widetilde{\Ydual}_{a|x}^\dagger$ and $\widetilde{\omega}\ge0$ achieving
\begin{align}
&\sum_{a,x}{\rm tr}\left(\widetilde{\Ydual}_{a|x}\gamma\right)p_{a|x} + \widetilde{\omega} = 2^{-{\rm SR}_\gamma({\bm\sigma})} < 1;\label{Eq:first constraint}\\
&\sum_{a,x}D(a|x,i)\widetilde{\Ydual}_{a|x}\ge0\;\forall\,i;\label{Eq:second constraint}\\
&\sum_{a,x}{\rm tr}\left(\widetilde{\Ydual}_{a|x}\gamma\right)p_{a|x}+\widetilde{\omega}\ge\sum_{a,x}{\rm tr}\left(\widetilde{\Ydual}_{a|x}\sigma_{a|x}\right)+1.\label{Eq:third constraint}
\end{align}
Importantly, the operators $\widetilde{\Ydual}_{a|x}$'s can be found by running SDP.
Now, by multiplying Eq.~\eqref{Eq:second constraint} by any $\eta_i\ge0$ with the condition $\sum_i{\rm tr}(\eta_i)=1$, summing over $i$, taking trace, and using Fact~\ref{Eq:LHSfact}, we obtain
\begin{align}\label{Eq:necessary second constraint}
\min_{{\bm\tau}\in{\bf LHS}({\bm\sigma})}\sum_{a,x}{\rm tr}\left(\widetilde{\Ydual}_{a|x}\tau_{a|x}\right)\ge\min_{{\bm\tau}'\in{\bf LHS}}\sum_{a,x}{\rm tr}\left(\widetilde{\Ydual}_{a|x}\tau'_{a|x}\right)\ge0.
\end{align}
Now, consider the Hamiltonians 
\begin{align}
H_{a|x}\coloneqq-(k_BT)|{\bf x}|\times\widetilde{\Ydual}_{a|x},
\end{align}
which, again, can be found by running SDP. 
Then we have
\begin{align}
\frac{\WdiffSA\left({\bm\sigma},{\bf H}\right)}{k_BT}& = \sum_{a,x}{\rm tr}\left[\widetilde{\Ydual}_{a|x}\left(p_{a|x}\gamma - \sigma_{a|x}\right)\right]\nonumber\\
&\ge1-\widetilde{\omega}\nonumber\\
&>\sum_{a,x}{\rm tr}\left(\widetilde{\Ydual}_{a|x}\gamma\right)p_{a|x}\nonumber\\
&\ge\sum_{a,x}{\rm tr}\left(\widetilde{\Ydual}_{a|x}\gamma\right)p_{a|x} - \min_{{\bm\tau}\in{\bf LHS}({\bm\sigma})}\sum_{a,x}{\rm tr}\left(\widetilde{\Ydual}_{a|x}\tau_{a|x}\right)\nonumber\\
&= \max_{{\bm\tau}\in{\bf LHS}({\bm\sigma})}\sum_{a,x}{\rm tr}\left[\widetilde{\Ydual}_{a|x}\left(p_{a|x}\gamma - \tau_{a|x}\right)\right]\nonumber\\
& = \max_{{\bm\tau}\in{\bf LHS}({\bm\sigma})}\frac{\WdiffSA\left({\bm\tau},{\bf H}\right)}{k_BT}.
\end{align}
We used Eq.~\eqref{Eq:useful2} in the first line.
The second line is due to Eq.~\eqref{Eq:third constraint}. 
In the third line, the strict inequality follows from Eq.~\eqref{Eq:first constraint}.
The fourth line can be obtained by using Eq.~\eqref{Eq:necessary second constraint}.
Finally, for ${\bm\tau}\in{\bf LHS}({\bm\sigma})$, we have ${\rm tr}(\tau_{a|x}) = p_{a|x}$ $\forall\,a,x$, which leads to the last line by using Eq.~\eqref{Eq:useful2}.
\end{proof}

To complete this section, we prove a fact that was used to show the strong duality in the proof of Result~\ref{Result:WorkExtraction}:
\begin{fact}\label{lemma: LHS sigma relint}
Let ${\bm\sigma}$ be a state assemblage with reduced state $\gamma$.
Suppose that $\gamma$ is full-rank.
Then there exist two state assemblages ${\bm\tau},{\bm\omega}$ in ${\bf LHS}({\bm\sigma})$ and $0<p<1$ such that
\begin{align}
{\bm\sigma} = \frac{{\bm\tau} + (p-1){\bm\omega}}{p}.
\end{align}
Hence, ${\bm\sigma}\in{\rm aff}({\bf LHS}({\bm\sigma}))$.
\end{fact}
\begin{proof}
To start with, let $\{\ket{i}\}_{i=0}^{d-1}$ be the eigenbasis of $\gamma$ ($d$ is  the system dimension).
Then $\ket{\Phi^+}_{AB}\coloneqq\sum_{i=0}^{d-1}(1/\sqrt{d})\ket{ii}_{AB}$ is a maximally entangled state in a bipartite system $AB$ with local dimension $d$.
Now, we define
\begin{align}
\ket{\gamma}_{AB}\coloneqq(\sqrt{d\gamma}\otimes\id_B)\ket{\Phi^+}_{AB}.
\end{align} 
One can check that it is a pure state whose marginals in $A$ and $B$ are both $\gamma$; that is, ${\rm tr}_A(\proj{\gamma}_{AB}) = \gamma = {\rm tr}_B(\proj{\gamma}_{AB})$.
For the given state assemblage ${\bm\sigma}$, one can use a spatial steering scenario in $AB$ to realise it as follows (see, e.g., Ref.~\cite{Uola2015PRL}):
\begin{align}
% $
\sigma_{a|x} = {\rm tr}_A\left[\proj{\gamma}_{AB}(E_{a|x}\otimes\id_B)\right]\quad\forall\,a,x,
% $
\end{align}
where $\{E_{a|x}\}_{a,x}$ is a collection of positive operator-valued measures in system $A$.
Now, consider another state assemblage ${\bm\omega} = \{\omega_{a|x}\}_{a,x}$ defined by
\begin{align}
\omega_{a|x}\coloneqq{\rm tr}_A\left[\left(\gamma_A\otimes(\id_B/d)\right)(E_{a|x}\otimes\id_B)\right]\quad\forall\,a,x.
\end{align}
Then we have, for every $a,x$,
\begin{align}
{\rm tr}(\omega_{a|x}) = {\rm tr}(\gamma E_{a|x}) = {\rm tr}(\sigma_{a|x}).
\end{align}
Namely, ${\bm\omega}\in{\bf LHS}({\bm\sigma})$.
For every $0<p<1$, we write
\begin{align}
\tau_{a|x}^{(p)}\coloneqq p\sigma_{a|x} + (1-p)\omega_{a|x} = {\rm tr}_A\left[\eta_{AB}(p)(E_{a|x}\otimes\id_B)\right],
\end{align}
where we define
\begin{align}
&\eta_{AB}(p)\coloneqq p\proj{\gamma}_{AB} + (1-p)\gamma_A\otimes(\id_B/d)\nonumber\\
&=(\sqrt{d\gamma}\otimes\id_B)\left(p\proj{\Phi^+}_{AB} + (1-p)(\id_{AB}/d^2)\right)(\sqrt{d\gamma}\otimes\id_B).
\end{align}
Note that $(\cdot)\mapsto\sqrt{\gamma}(\cdot)\sqrt{\gamma}$ is an ${\rm LF_1}$ filter in $A$ (with success probability $1/d$), which {\em cannot} map a separable state to an entangled one.
Hence, by selecting a sufficiently small (while still positive) $p$, the state $p\proj{\Phi^+}_{AB} + (1-p)\id_{AB}/d^2$ will be separable, and hence $\eta_{AB}(p)$ will be separable, too.
Since separable states can only induce LHS state assemblages, we conclude that ${\bm\tau}^{(p)}\in{\bf LHS}$ for small enough $p$.
Finally, by construction, we also have ${\rm tr}(\tau_{a|x}^{(p)}) = {\rm tr}(\sigma_{a|x})$ for every $a,x$.
This means that, for some $0<p<1$ and ${\bm\tau}^{(p)},{\bm\omega}\in{\bf LHS}({\bm\sigma})$,
\begin{align}\label{Eq: aff LHS sigma computation 0001}
{\bm\sigma} = \frac{{\bm\tau}^{(p)} + (p-1){\bm\omega}}{p}.
\end{align}
Consequently, we have ${\bm\sigma}\in{\rm aff}({\bf LHS}({\bm\sigma}))$.
\end{proof}
Finally, as a corollary, we have 
\begin{align}\label{Eq:sigma lambda in aff LHS sigma}
{\bm\sigma}^{(\SRvar)}\in{\rm aff}({\bf LHS}(\bm\sigma))\quad\forall\,\SRvar\ge1,
\end{align}
where the state assemblage ${\bm\sigma}^{(\SRvar)}$ is defined in Eq.~\eqref{Eq:sigma^lambda_def}.
To see this, using Eq.~\eqref{Eq: aff LHS sigma computation 0001}, we can write
\begin{align}
{\bm\sigma}^{(\SRvar)} = \frac{1}{p\SRvar}{\bm\tau}^{(p)} + \frac{p-1}{p\SRvar}{\bm\omega} + \frac{\SRvar-1}{\SRvar}\{p_{a|x}\gamma\}_{a,x},
\end{align}
where ${\bm\tau}^{(p)},{\bm\omega},\{p_{a|x}\gamma\}_{a,x}$ are all in ${\bf LHS}({\bm\sigma})$ and $1/(p\SRvar) + (p-1)/(p\SRvar) + (\SRvar-1)/\SRvar=1$.
We then conclude the desired claim by using the definition of affine hull [see, e.g., Eq.~\eqref{Eq:affineLHS}].
\\

\subsection{\CYtwo{Supplemental Material IX: Work extraction task induced by a steering inequality}}

\CYtwo{
We detail the example given in Appendix C, which aims to illustrate how to certify ${\rm SR}_\gamma({\bm\sigma})>0$ via the work extraction advantage shown in Result~\ref{Result:WorkExtraction} in the main text. 
To start with, consider the single-qubit state assemblage ${\bm\sigma}^{\rm Pauli}$ defined by Eqs.~\eqref{Eq: sigma Pauli example} and~\eqref{Eq: sigma Pauli example 2} in Appendix C:
\begin{equation}
\begin{aligned}
&\sigma_{0|0}^{\rm Pauli}\coloneqq\proj{+}/2, ~~\sigma_{1|0}^{\rm Pauli}\coloneqq\proj{-}/2,\\
&\sigma_{0|1}^{\rm Pauli}\coloneqq\proj{0}/2, ~~\sigma_{1|1}^{\rm Pauli}\coloneqq\proj{1}/2,
\end{aligned}
\label{Eq_App_Pauli_Assem}
\end{equation}
where $\ket{\pm}\coloneqq(\ket{0}\pm\ket{1})/\sqrt{2}$.
% In other words, 
${\bm\sigma}^{\rm Pauli}$ is prepared by applying projective measurements of Pauli $X,Z$ on the thermal state $\gamma = \id/2$. 
Here, Pauli $X,Z$ are observables defined by
\begin{align}
X \coloneqq \proj{+} - \proj{-}\quad\&\quad Z \coloneqq \proj{0} - \proj{1}.
\end{align}
In this case, we have $P(a,x) = {\rm tr}(\sigma_{a|x}^{\rm Pauli}) = 1/2$ for every $a,x=0,1$.
Hence, certifying ${\rm SR}_\gamma({\bm\sigma})>0$ via Result~\ref{Result:WorkExtraction} in the main text is equivalent to demonstrating the following strict inequality for some $H_{a|x}$'s [i.e., Eq.~\eqref{Eq:12} in the main text]:
\begin{align}\label{Eq:goal}
\max_{{\bm\tau}\in{\bf LHS}\left({\bm\sigma}^{\rm Pauli}\right)}\sum_{a,x}\frac{\Wdiff(\hat{\tau}_{a|x},H_{a|x})}{2}<\sum_{a,x}\frac{\Wdiff(\hat{\sigma}_{a|x}^{\rm Pauli},H_{a|x})}{2}.
\end{align}
Now, we seek experimentally feasible Hamiltonians $H_{a|x}$'s by considering {\em steering inequalities}, which are commonly used to certify steerability of state assemblages. There are many forms of steering inequalities, and we use the one shown in Ref.~\cite{Pusey2013} (originated from Ref.~\cite{Cavalcanti2009}):
\begin{align}\label{Eq:LHS bound}
\sum_{a,x}{\rm tr}(F_{a|x}\tau_{a|x})\leq \sqrt{2}\quad\forall\,{\bm\tau}\in{\bf LHS},
\end{align}
with
\begin{align}
F_{0|0}=-F_{1|0}=X\quad\&\quad F_{0|1}=-F_{1|1}=Z.
\end{align}
The value of $\sqrt{2}$ is the maximum that can be achieved by an LHS model. With the quantum strategy of Eq.~\eqref{Eq_App_Pauli_Assem}, we can reach the maximal quantum violation of the value $2$; namely,
\begin{align}\label{Eq:Q bound}
\sum_{a,x}{\rm tr}(F_{a|x}\sigma^{\rm Pauli}_{a|x}) = 2.
\end{align}
We thus define the following Hamiltonians for some parameters $0<\delta<\infty$ and temperature $0<T<\infty$:
\begin{align}\label{Eq: H def}
H_{a|x}^{\rm Pauli}\coloneqq k_BT\delta\times F_{a|x}\quad\forall\,a,x.
\end{align}
In other words, we have, for $a=0,1$,
\begin{align}
&H_{a|0}^{\rm Pauli}=(-1)^ak_BT\delta\times X;\\ &H_{a|1}^{\rm Pauli}=(-1)^ak_BT\delta\times Z,
\end{align}
which are identical to Eq.~\eqref{Eq:13} in Appendix C.
They are Hamiltonians induced by a steering inequality.
By selecting a suitable parameter $\delta$, they are realisable experimentally (see also the example in Appendix C).
Now, note that
\begin{align}
{\rm tr}\left(e^{-\frac{H_{a|x}^{\rm Pauli}}{k_BT}}\right) = e^\delta + e^{-\delta} = 2\cosh\delta\quad\forall\,a,x.
\end{align}
Using Eq.~\eqref{Eq:Wdiff formula}, we have, for every state $\rho$,
\begin{align}\label{Eq: Wdiff Pauli}
\Wdiff(\rho,H_{a|x}^{\rm Pauli}) &= {\rm tr}(H_{a|x}^{\rm Pauli}\rho) + k_BT\ln{\rm tr}\left(e^{-\frac{H_{a|x}^{\rm Pauli}}{k_BT}}\right) - k_BT\ln 2\nonumber\\
&= {\rm tr}(H_{a|x}^{\rm Pauli}\rho) +k_BT\ln\left(\cosh\delta\right).
\end{align}
Hence, by combining everything, we obtain
\begin{align}
&\max_{{\bm\tau}\in{\bf LHS}\left({\bm\sigma}^{\rm Pauli}\right)}\sum_{a,x}\frac{\Wdiff(\hat{\tau}_{a|x},H_{a|x}^{\rm Pauli})}{2}\nonumber\\
&=2k_BT\ln\left(\cosh\delta\right) + \max_{{\bm\tau}\in{\bf LHS}\left({\bm\sigma}^{\rm Pauli}\right)}\sum_{a,x}{\rm tr}(\tau_{a|x}H_{a|x}^{\rm Pauli})\nonumber\\
&\le 2k_BT\ln\left(\cosh\delta\right) + \max_{{\bm\tau}'\in{\bf LHS}}\sum_{a,x}{\rm tr}(\tau'_{a|x}H_{a|x}^{\rm Pauli})\nonumber\\
&\le k_BT \left[\sqrt{2}\delta + 2\ln\left(\cosh\delta\right)\right]\nonumber\\
&<2k_BT \left[\delta + \ln\left(\cosh\delta\right)\right]=\sum_{a,x}\frac{\Wdiff(\hat{\sigma}_{a|x}^{\rm Pauli},H_{a|x}^{\rm Pauli})}{2},
\end{align}
which is exactly what we aim at [i.e., Eq.~\eqref{Eq:goal}].
Here, the second line is due to Eq.~\eqref{Eq: Wdiff Pauli}, and the third line is due to a larger maximisation range. 
Using Eq.~\eqref{Eq: H def}, we obtain the fourth line via Eq.~\eqref{Eq:LHS bound} and the fifth line via Eq.~\eqref{Eq:Q bound}.
Also note that, for every $a,x$, we have $\hat{\sigma}_{a|x}^{\rm Pauli}/2 = \sigma_{a|x}^{\rm Pauli}$ and $\hat{\tau}_{a|x}/2 = \tau_{a|x}$ for ${\bm\tau}\in{\bf LHS}\left({\bm\sigma}^{\rm Pauli}\right)$.
We remark that the fourth and fifth lines give the classical bound [Eq.~\eqref{Eq:example classical bound}] and quantum bound [Eq.~\eqref{Eq:example quantum bound}] in Appendix C.

Finally, recall that we set $\gamma=\id/2$.
By using Eq.~\eqref{Eq:useful2} and the fact that $X,Z$ are traceless, we obtain (${\bf H}^{\rm Pauli}\coloneqq\{H_{a|x}^{\rm Pauli}\}_{a,x}$)
\begin{align}
&\WdiffSA({\bm\sigma}^{\rm Pauli},{\bf H}^{\rm Pauli}) = \sum_{a,x}\frac{{\rm tr}\left(H_{a|x}^{\rm Pauli}\sigma_{a|x}^{\rm Pauli}\right)}{2}\nonumber\\
& = \frac{k_BT\delta}{2}\times\sum_{a,x}{\rm tr}\left(F_{a|x}\sigma_{a|x}^{\rm Pauli}\right) = k_BT\delta,
\end{align}
where we have used Eqs.~\eqref{Eq:Q bound} and~\eqref{Eq: H def}.
Hence, we have $\WdiffSA({\bm\sigma}^{\rm Pauli},{\bf H}^{\rm Pauli})\ge k_BT\eta$ if and only if $\delta\ge\eta$.
}

\bibliography{Ref.bib}

\end{document}